\documentclass[12pt]{article}         
\pdfoutput=1
\usepackage{pu_ece_report} 
\usepackage{fancybox}
\usepackage{fancyheadings}
\usepackage{graphicx} 
\usepackage{subfigure}

\makeatletter

\let\proof\@undefined
\let\endproof\@undefined
\makeatother

\usepackage{amsmath,amsthm,amssymb}

\usepackage{bm}
\usepackage{comment}
\usepackage{algorithm}
\usepackage{algorithmic}
\usepackage{subfigure}
\usepackage{multicol}
\usepackage{colortbl}
\usepackage{natbib}
\def\bmath#1{\boldsymbol{#1}}
\def\Gn{\mathcal{G}_n}                
\def\th{\bmath{\theta}}               
\def\fG{\bmath{f}\left(G\right)}      
\def\x{\bmath{x}}                     
\def\xstar{\x^\star}                  
\def\thxstar{\hat{\th}\left(\xstar\right)}  
\def\Gstar{G^\star}                   
\newcommand{\argmax}{\operatornamewithlimits{argmax}}   
\def\Sp{\mathbb{S}^{p-1}}             
\def\feat{\mathcal{P}}                
\def\CH{\mathcal{C}}                  
\def\V{\mathcal{V}}                   
\def\nb{\mathcal{B}}                  
\def\emb{\bmath{e}}                   
\def\SB{\mathcal{S}\left(\nb\right)}  
\def\inner#1#2{\left<#1,#2\right>}    
\def\H{\mathcal{H}}                   
\def\ELSRGM{\texttt{ELSRGM}}

\newtheorem{theorem}{Theorem}

\begin{document}

\buecedefinitions%
        {Generating Similar Graphs From Spherical Features}
        {Degeneracy and Exponential Locally Spherical Random Graph Model}
        {Dalton Lunga and Sergey Kirshner}
        {{\tt \{dlunga,skirshne\}@purdue.edu}}
        {May 12 2011}
        {11-01} 


\buecereporttitlepage


\pagenumbering{roman}
\setcounter{page}{1}
\buecereportsummary{We propose a novel model for generating graphs similar to a given example graph.  Unlike standard approaches that compute features of graphs in Euclidean space, our approach obtains features on a surface of a hypersphere.  We then utilize a von Mises-Fisher distribution, an exponential family distribution on the surface of a hypersphere, to define a model over possible feature values.  While our approach bears similarity to a popular exponential random graph model (ERGM), unlike ERGMs, it does not suffer from degeneracy, a situation when a  significant probability mass is placed on unrealistic graphs.  We propose a parameter estimation approach for our model, and a procedure for drawing samples from the distribution.  We evaluate the performance of our approach both on the small domain of all $8$-node graphs as well as larger real-world social networks.}

\tableofcontents\bueceemptypage
\listoffigures\bueceemptypage
\listoftables\bueceemptypage

\buecereportheaders

\pagenumbering{arabic}
\setcounter{page}{1}

\section{Introduction}
Increasingly, many domains produce data sets containing relationships that are conveniently represented by networks, e.g., systems sciences (the Internet), bioinformatics (protein interactions), social domains (social networks).  As researchers in these areas are developing models and tools to analyze the properties of networks, they are hampered by few samples available to evaluate their approaches.  This gives rise to a problem of generating more network samples that can be viewed as drawn from the same population as the given network.

While there are a number of possible approaches to this problem, perhaps the most well-studied model is the exponential random graph models (ERGMs, or in social network literature, $p^\star$), an exponential family class of models, matching the statistics over the set of possible networks to the statistics of the network in question \citep[e.g.,][]{WassermanPattison}.  This and similar approaches have a long history as they generalize the $p_{1}$ \cite{HollandLeinhardt1981} and Markov random graph \cite{FrankStrauss} models first developed in the social network literature. While such approaches are intuitive and have nice properties, they also suffer from the issue of {\em degeneracy} \citep{Handcock2003b,Rinaldo2009}, which is manifested in the instability of parameter estimation, and in placing most probability mass of resulting distributions on unrealistic graphs (e.g., empty graph or complete graphs). As a result, these approaches are not suitable for the purpose of generating graphs similar to the given one.

This paper contains two contributions.  First, we zero-in on the issue of degeneracy and discover that its cause is related to the geometry of the set of feature vectors and the number of graphs mapped to each feature vector (thought of as feature vector weights).  By augmenting feature vectors with logarithms of their corresponding weights, we show that only graphs with such augmented feature vectors on the relative boundary of the resulting extended convex hull can become modes of any exponential random graph model, explaining why unrealistic graphs (which are on the relative boundary) often get large probability masses. Second, using the insight of the observation above, we propose a novel random graph model which is based on embedding the features of graphs onto a surface of a hypersphere.  Since a spherical surface is a relative boundary of the sphere's convex hull, all of the feature vectors would then belong to the relative boundary of the convex hull and could potentially serve as modes of the corresponding distributions.  This in turn helps to avoid the degeneracy issues which plague ERGMs.

Our proposed approach makes use of {\em spherical} features obtained by embedding possible graphs onto a surface of a sphere \citep{Wilson2010}, and then approximating the distribution of the resulting spherical feature space with a von Mises-Fisher distribution.  Since the space of all possible graphs is too large to consider for embedding, we consider determining the embedding function based only on the neighborhood around the given graph thus resulting in a locally spherical embedding of the set of graphs. The main benefit of our approach is that it fixes the issue of degeneracy, with the mode of the distribution over the spherical feature vector coinciding with the features of the given graph.  An additional advantage of our approach is that its parameter estimation procedure does not require cumbersome maximum entropy approaches used with ERGMs.

We start by revisiting the ERGM model and presenting insights on why this model often fails to generate realistic graphs (Section \ref{sec:ergm}).  We then propose our alternative approach, exponential locally spherical random graph model (\ELSRGM, Section \ref{sec:elsrgm}) and evaluate it on both synthetic and realistic graphs while comparing it to ERGMs (Section \ref{sec:experiments}).  We conclude the paper (Section \ref{sec:conclusion}) with ideas for future work.

\section{ERGMs}\label{sec:ergm}
Overall, we are interested in probabilistic approaches for generating graphs similar to the given one.  We consider the case of simple (unweighted, no self-loops) undirected graphs $G=\left(V,E\right)\in \Gn$.\footnote{{\scriptsize Our approach extends to directed graphs as well.}}  Where $n=\left|V\right|$ is the number of vertices; $E\in\left\{0,1\right\}^{n\times n}$ is a symmetric binary adjacency matrix with zeros on its diagonal, with $e_{ij}=1$ iff there is an edge from $v_i$ to $v_j$, and $e_{ij}=0$ otherwise, where $v_{i}\in V$ and $e_{ij}\in E$.  There are  $\left|\Gn\right|=2^{\binom{n}{2}}$ possible labeled graphs with $n$ vertices, a finite but often prohibitively large number even for fairly small $n$.

We first consider the well-studied exponential random graph model (ERGM, also known as $p^\star$) as a starting point for our approach.
\subsection{ERGM Definition}
In the area of social network analysis, scientists are often interested in specific features of networks, and some of the state-of-the-art models explicitly use them to define functions of network sub-structures.  We will denote the vector of these functions by $\bmath{f}:\Gn\to\mathbb{R}^d$.  Among the examples of such features used by social scientists, we have the number of edges, the number of triangles, the number of $k$-stars, etc.:
\begin{eqnarray}
f_{edge}\left(G\right) &=& \sum\sum_{1\leq<i<j\leq n}e_{ij};\nonumber\\
f_{\triangle}\left(G\right) &=& \operatornamewithlimits{\sum\sum\sum}_{1\leq i<j<k\leq n}e_{ij}e_{ik}e_{jk};\nonumber\\
f_{k\star}\left(G\right) &=& \sum_{i=1}^n\operatornamewithlimits{\sum\dots\sum}_{\tiny \begin{array}{c}1\leq i_1<\dots\leq i_k\\i_1\neq i,\dots,i_k\neq i\end{array}}\prod_{j=1}^ke_{ii_j}.
\label{eqn:features}
\end{eqnarray}

The subgraph patterns corresponding to features in equation (\ref{eqn:features}) are shown in Figure \ref{fig:countstatistics}.

\begin{figure}[!htbp]
\centering
\includegraphics[width=1.17in]{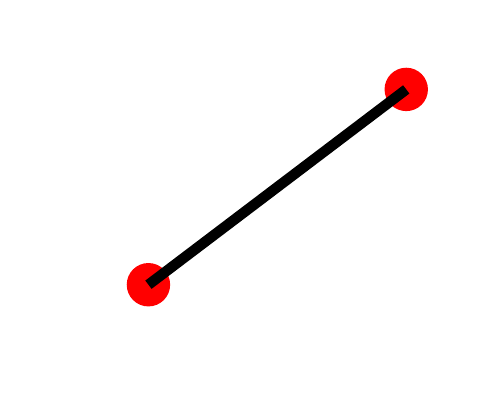}\hspace*{0.5cm}
\includegraphics[width=1.45in]{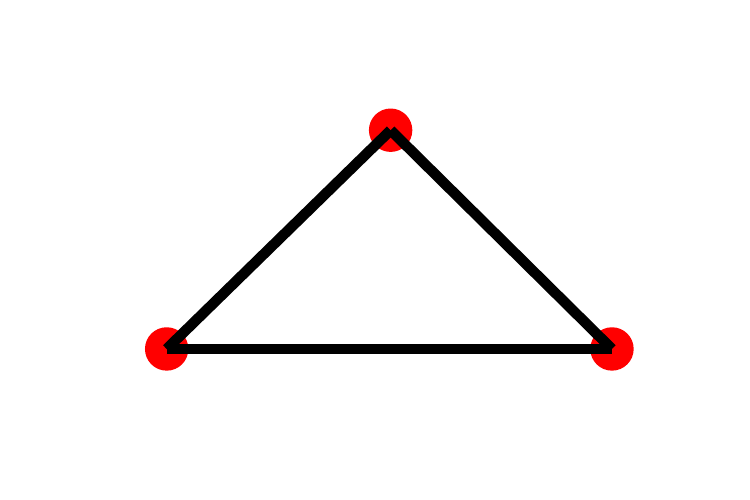}\hspace*{0.5cm}
\includegraphics[width=1.45in]{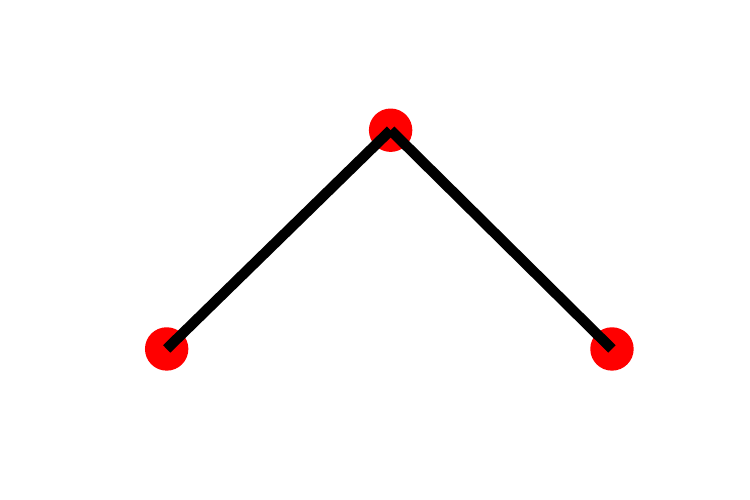}\hspace*{0.5cm}
\includegraphics[width=1.35in]{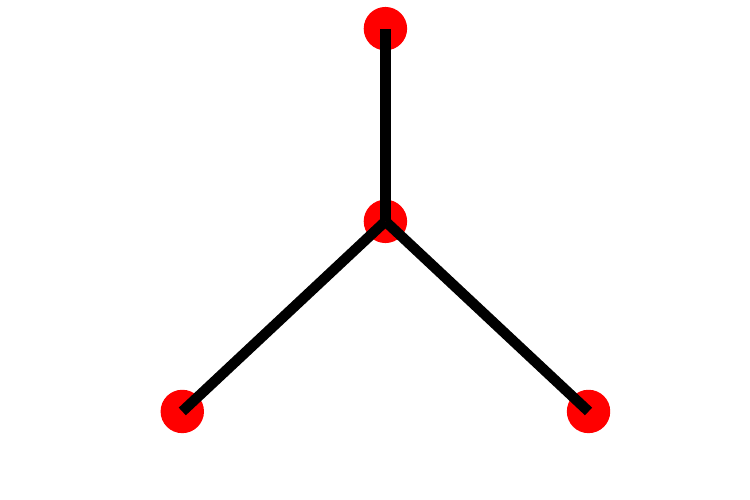}
\caption{Simple typical subgraph configurations for undirected graphs. From \textit{left} to \textit{right}: {\tt edge}, {\tt triangle}, {\tt two-star} and {\tt three-star} configurations.}
\label{fig:countstatistics}
\end{figure}

A commonly used probabilistic model over $\Gn$ is an exponential family model that uses the expectation of $\bmath{f}$ as a vector of sufficient statistics.  The distribution over $\Gn$ can be parameterized in the form
\begin{equation}
P\left(G\vert \th\right)=\frac{\exp\inner{\th}{\fG}}{Z\left(\th\right)},\ G\in\Gn,
\label{eqn:ergm}
\end{equation}
where $\th\in\mathbb{R}^d$ is the vector of natural model parameters, $\fG$ is a vector of features of $G$, and \[Z\left(\th\right)=\sum_{G\in\Gn}\exp\inner{\th}{\fG}\]
is the partition function.

Let $\feat=\left\{\fG: G\in\Gn\right\}$ denote the set of possible feature-vector values.  Under ERGMs, the space $\Gn$ of possible graphs is coarsened into a much smaller set $\feat$ of possible features.  The distribution in \eqref{eqn:ergm} can also be viewed as a distribution over $\feat$ with all of the graphs mapped to the same value of $\bmath{f}$ assigned the same probability mass.  Let $w\left(\x\right)=\left|\left\{G\in\Gn:\fG=\x\right\}\right|$ be the number
of graphs corresponding to a feature value $\x\in\feat$.  (We will refer to $w\left(\x\right)$'s as {\em weights}.)  The distribution $P\left(G\vert\th\right)$ can also be extended to the set $\feat$,
\begin{equation}
\begin{split}
P\left(\x\vert \th\right) &= \frac{1}{Z\left(\th\right)}w\left(\x\right)\exp\inner{\th}{\x},\ \x\in\feat,\\
Z\left(\th\right) &= \sum_{\x\in\feat}w\left(\x\right)\exp\inner{\th}{\x},
\end{split}\label{eqn:ergm weighted}
\end{equation}
also an exponential family distribution.  We let $\CH\subseteq\mathbb{R}^d$ denote a convex hull of $\feat$.  Since $\Gn$ is finite, so would be $\feat$, and $\CH$ would be a polytope in $\mathbb{R}^d$.  Using a well-known result from the theory of exponential families \cite{Barndorff-Nielsen1978}, if $\xstar\in\mathbb{R}^d$ is a vector of the sufficient statistics (e.g., for a single example graph $\Gstar$, $\xstar=\bmath{f}\left(\Gstar\right)$), a maximum likelihood estimate (MLE) $\thxstar$ satisfies
\begin{eqnarray}
    E_{P\left(G\vert\thxstar\right)}\fG = E_{P\left(\x\vert\thxstar\right)}\x = \xstar,
    \label{eqn:maxent}
\end{eqnarray}
and it exists, and the corresponding distribution $P\left(\cdot\vert \thxstar\right)$ is unique as long as $\xstar\in \mbox{rint}\left(\CH\right)$, the relative interior of the convex hull of the set of possible feature vectors $\feat$.  In essence, ERGMs are designed to preserve mean features of the observed graph, a very intuitive and often desirable property.  However, it is also well-known that the estimation of $\thxstar$ is an extremely cumbersome task, complicated by the fact that the exact calculation of $Z\left(\th\right)$ is intractable, and approximate approaches (e.g., pseudo-likelihood, MCMC) are employed instead \cite{Hunter2008a}.

\subsection{Degeneracy in ERGMs}

ERGMs often suffer from the problem of {\em degeneracy} \cite{Handcock2003b}.  There are two types of degeneracy usually considered.
The first type occurs when the MLE estimate $\thxstar$ either does not exist,
or the MLE estimation procedure does not converge due to numerical instabilities \cite{Handcock2003b,Rinaldo2009}.

The second type of degeneracy happens when $\thxstar$ can be reliably estimated, but the resulting probability distribution places significant probability mass (or virtually {\em all} of probability mass) on unrealistic graphs, e.g., empty or complete graphs. This type of degeneracy can be considered from another viewpoint; that is the mode of ERGM corresponding to the $\thxstar$ may be placed on $\x\in\feat$ very different from $\xstar$.  This is an undesirable property as there is little justification for placing large probability mass over a region away from observed feature vectors while placing little mass over the observed example.

For an illustration, consider the set $\mathcal{G}_8$ (the set of all possible simple undirected graphs with $8$ nodes) with feature vectors $\fG=\left(f_{edge}\left(G\right),f_{\triangle}\left(G\right)\right)$.  In this case, Figure \ref{fig:g8distributionanddiameter} displays (left plot) the support space $\feat$ consists of $228$-({\em edge,triangle}) pair-statistics for all possible $8$-node undirected graphs. The right plot of Figure \ref{fig:g8distributionanddiameter}, displays a diameter distribution highlighting that graphs with different topologies map to same {\em edge-triangular} feature pairs. The diameter was computed by observing that each feature-pair could be viewed as a cell of graphs with same feature count. This can be accomplished by first computing a perturbation graph whose nodes are all non-isomorphic graphs of $8$-nodes. A graph edit distance is applied on the set $\mathcal{G}_8$ containing $12346$-non isomorphic graphs (computed using {\tt nauty} \cite{McKay1981}). To compute the diameter for each cell given the perturbation graph, one only has to identify the graphs mapping to that cell and extract the maximum number of edges for any pair of graphs with their feature counts mapping to that cell. Table \ref{tab:ComplexityOfERGModelsOnCanonicalSmallGraphSizesOfNNodes} shows the complexity of small sized graph spaces.

\begin{table*}[!htbp]
       \caption{Complexity of graph spaces for fixed number of nodes.}
          \centering
        \begin{tabular}{lccc}
		\hline
		n& number of edges & all graphs are $2^{n\choose 2}$& non-isomorphic\\
			\hline
			7 & 21 & 2,097,152 &1,044\\
			\textcolor[rgb]{0.00,1.00,0.00}{8} & \textcolor[rgb]{0.00,1.00,0.00}{28} & \textcolor[rgb]{0.00,1.00,0.00}{268,435,456} &\textcolor[rgb]{0.00,1.00,0.00}{12,346}\\
			9 & 36 &  68,719,476,736&274,668\\
			10 & 45 & 35,184,372,088,832 &12,005,168\\
			11 & 55 & 3,602,879,701,896,397 &1,018,997,864\\
			12 & 66 & 7,378,697,629,483,821,000 &165,091,172,592\\
			\hline
		\end{tabular}
	\label{tab:ComplexityOfERGModelsOnCanonicalSmallGraphSizesOfNNodes}
\end{table*}

\begin{figure*}[!htbp]
    \centering
    \begin{tabular}{cc}
      \hspace{-0.7cm}\includegraphics[scale=0.58]{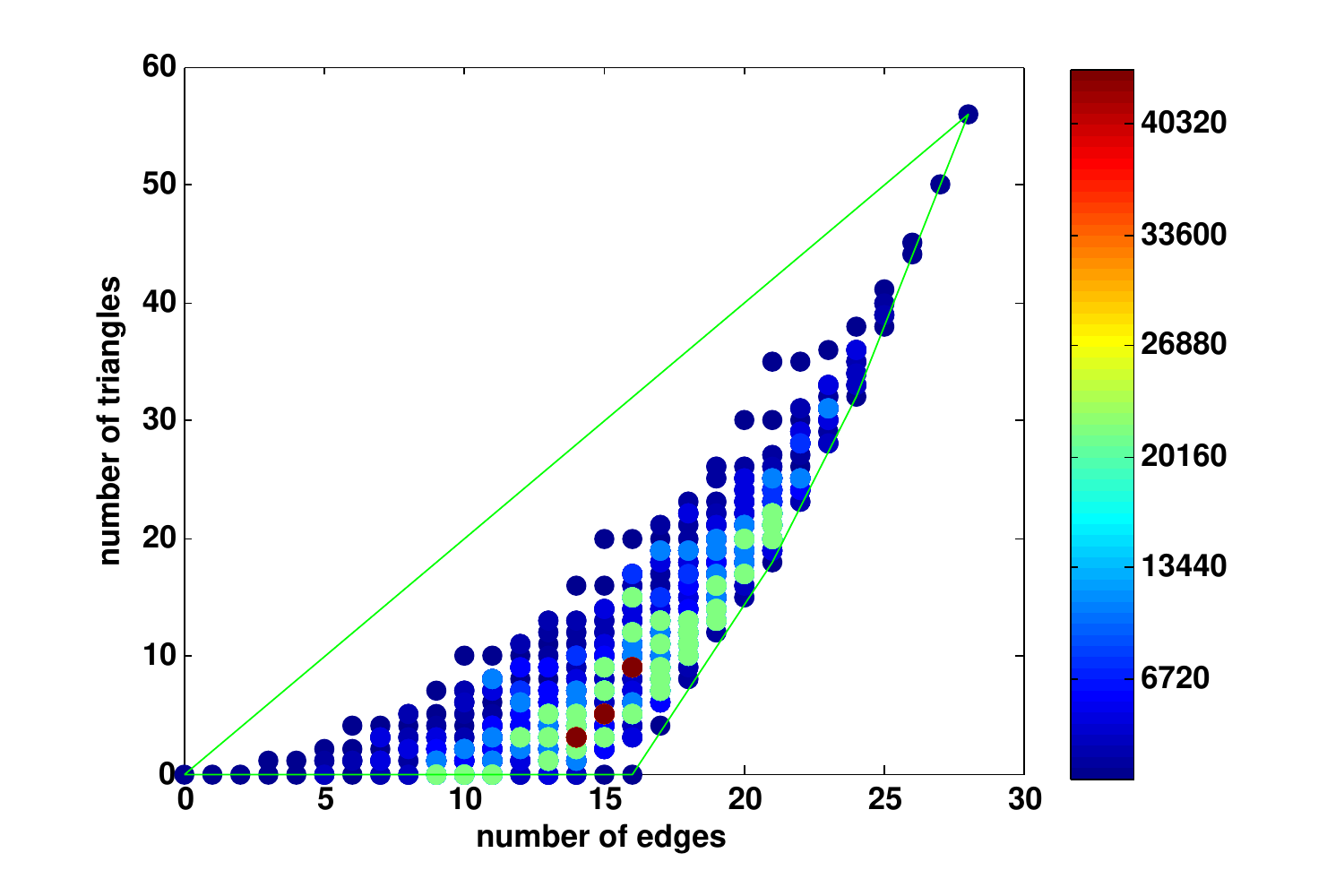}&\hspace{-1.2cm}\includegraphics[scale=0.58]{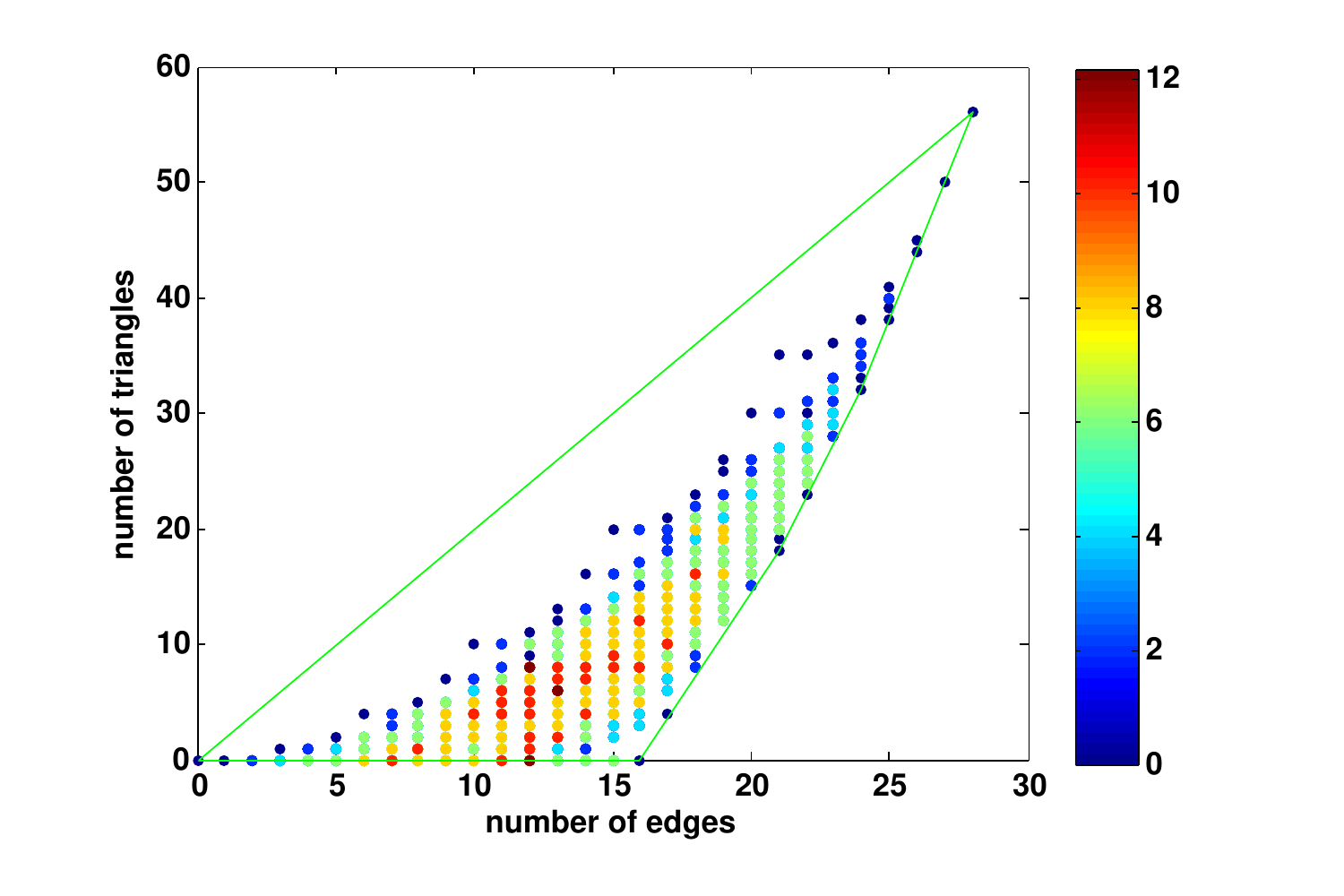}
    \end{tabular}
    \caption{{\em left:} Distribution of $8$-node graphs. {\em right:} Distribution of cell diameter; where we consider each feature-pair as a cell and compute the graph edit distance between each pair of graphs with feature counts mapping to that cell (feature-pair).}
    \label{fig:g8distributionanddiameter}
\end{figure*}

 Figure \ref{fig:fpmodes} shows the support $\feat$ for the distribution over feature-vectors (all circles), and its convex hull $\CH$ (boundary in green).  Only a small subset of the example feature pairs $\xstar=\bmath{f}\left(\Gstar\right)$ result in distributions over $\feat$ with modes coinciding the example feature vector $\xstar$ (red discs with black borders).   ERGMs estimated from other feature vectors will thus generate graphs with different features from the example graph as shown in Figure \ref{fig:degeneracy}.

\begin{figure*}[ht]
\begin{center}
\includegraphics[width=1.5in]{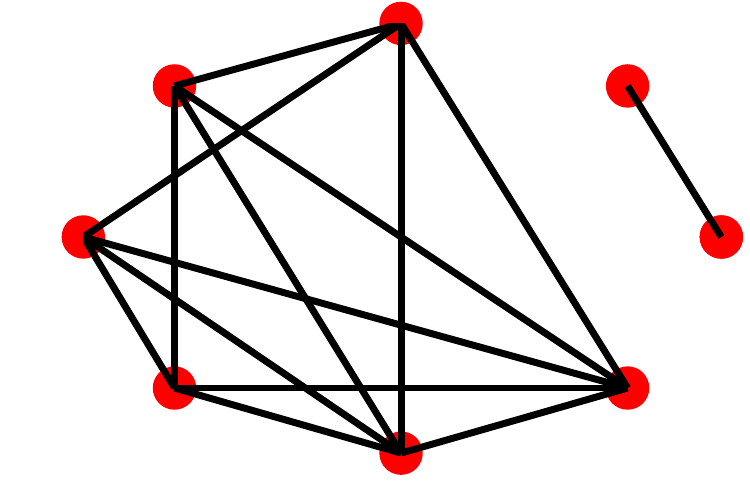}\\
\vspace{0.1cm}
\includegraphics[scale=0.9]{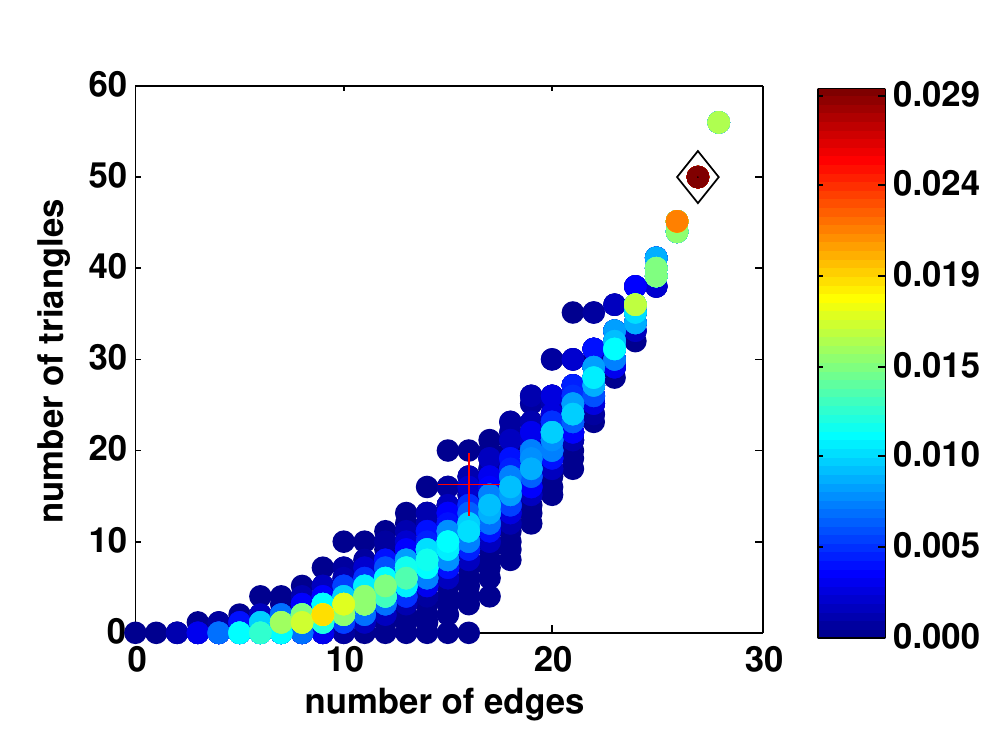}
\caption{A degenerate ERGM specified by edge-triangle pair for an $8$-node graph({\em top figure}). Colorcoded is the pmf over the edge-triangle space for an estimated MLE $\theta_{MLE}=(-0.992,0.617)$. "+" indicates the observed feature and its mean, while the "$\diamond$" shape indicates the ERGM mode.}
\label{fig:degeneracy}
\end{center}
\end{figure*}

It is this behavior that we are most concerned with and are trying to address in this paper, and from now on when we mention degeneracy, unless otherwise noted, we mean the second type of degeneracy.

\begin{figure*}[ht]
\begin{center}
\includegraphics[scale=0.7]{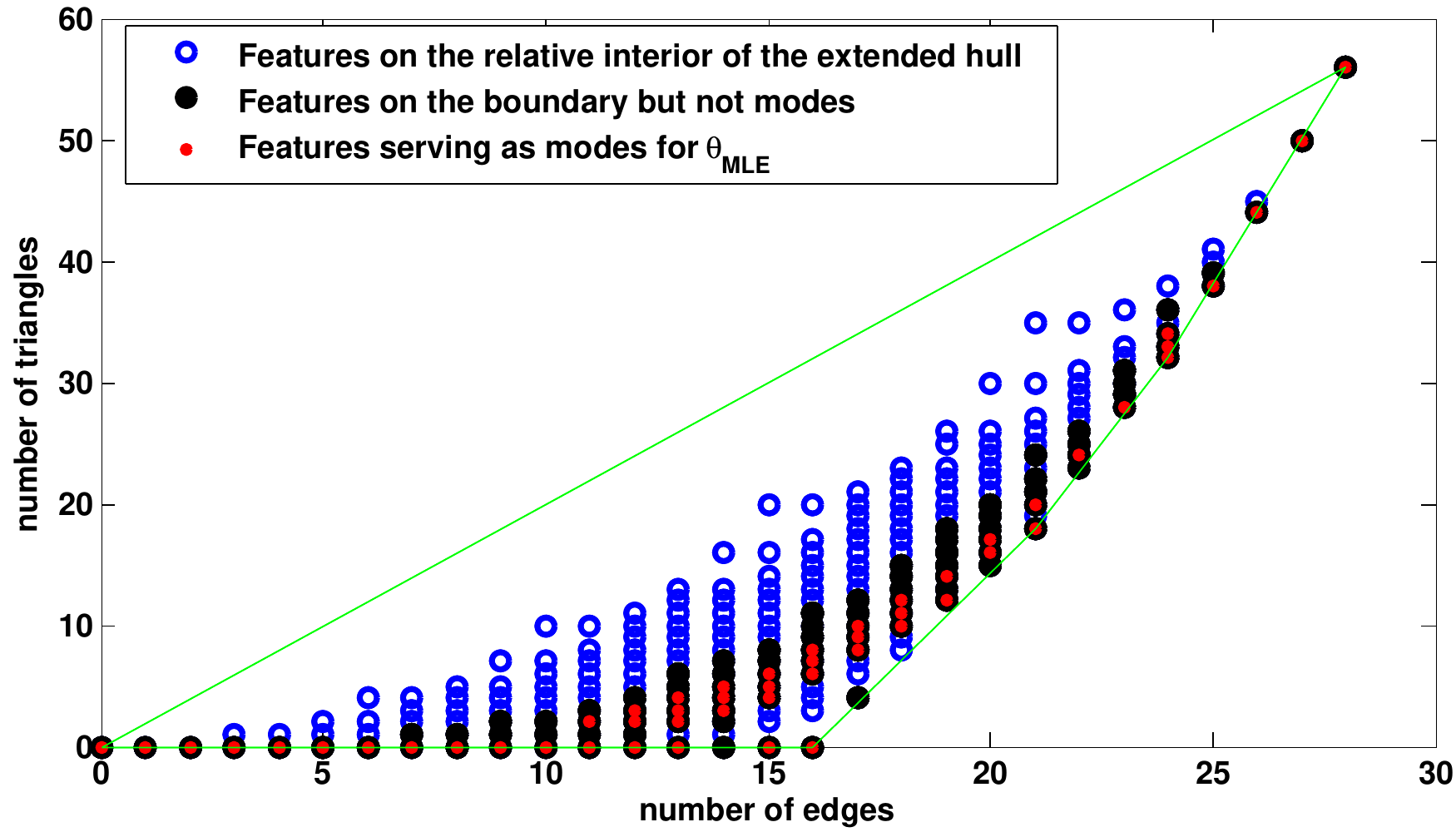}
\caption{Solid points show mode placement on the edge-triangle feature pairs that form the vertices of the extended hull as a result of Theorem \ref{thm:max inner product}. Overlayed on solid points are modes obtained from using the estimated $\hat{\theta}_{mle}$ of equation (\ref{eqn:ergm weighted}).}
\label{fig:fpmodes}
\end{center}
\end{figure*}

\begin{figure*}[ht]
\begin{center}
\includegraphics[scale=0.7]{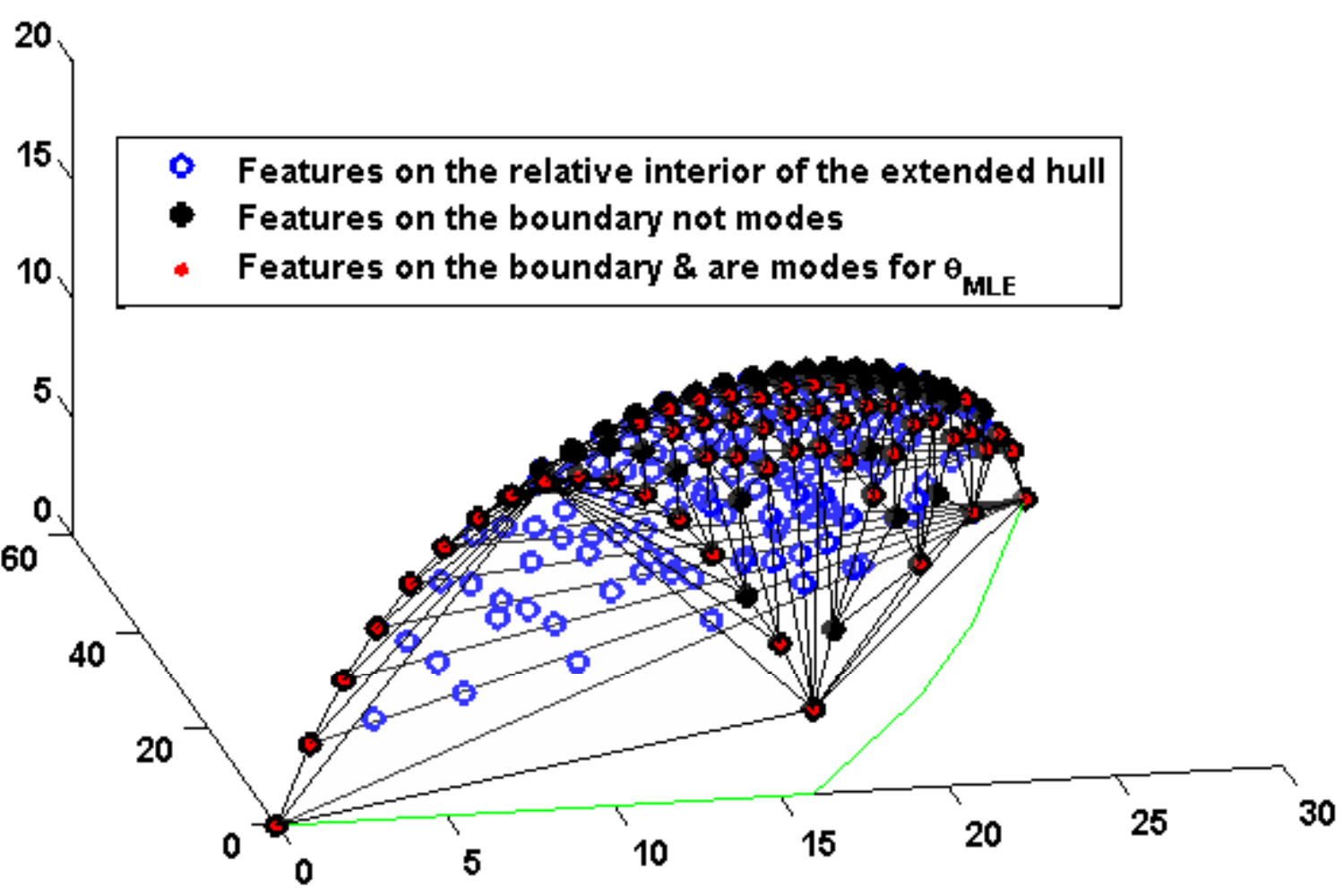}
\caption{View of mode placement on facets of the extended 3D convex hull. Each mode forms a vertex of a given triangle (denoted by black lines) for a facet of the convex hull.}
\label{fig:fpmodes3D}
\end{center}
\end{figure*}

\subsubsection{Approaches to Fixing Degeneracy in ERGMs}
Recent approaches in literature have focused on a more flexible specification of the model in equation \eqref{eqn:ergm weighted}. This has been achieved by a mixed set of feature statistics which includes the node attributes of a given graph \cite{Handcock2008}. Instead of using only the structural properties ({\em edges, triangles etc.}) of the graph, other attributes (e.g., {\em gender,race,age,etc}) are included as covariates for the model. This approach, however, does not address degeneracy in general as one has to know what set of features and attributes to choose in attempting to minimize degeneracy \cite{Handcock2008}. Such a task demands domain knowledge to accurately specify a reasonable model.  Another approach makes use of the geometrically weighted edgewise shared partner, the geometrically weighted dyadic shared partner, and the geometrically weighted degree network statistics as new statistics for the model in (\ref{eqn:ergm}). These specifications when parameterized have been shown to lead to curved exponential families \cite{Hunter2008a}.  However, the difficulty with this approach is not only in the parameter estimation of the resulting curved exponential model, but also on having to avoid other graph features that may be dependent to these specifications as that will result in degeneracy \cite{Hunter2006}.

\subsection{Why Are Unrealistic Graphs Likely under ERGMs?}
It may not be entirely surprising that the probability mass over support $\feat$ does not concentrate around the expected value of $P$ as the mode and the mean of the exponential family distributions do not always match.  This does not however imply that the mode is placed on features corresponding to unrealistic graphs.  To investigate the degeneracy further, we first introduce the following result from convex analysis.
\begin{theorem}\label{thm:max inner product}
Suppose $\CH\subset \mathbb{R}^d$ is a full-dimensional bounded convex polytope with a finite set of vertices $\V=\left\{\x^1,\dots,\x^L\right\}$.  Then for any $\th\in\mathbb{R}^d\setminus\left\{\bmath{0}\right\}$,
\begin{equation}
\mathcal{M}_{\th}=\argmax_{\x\in\CH}\inner{\th}{\x}\subset\mbox{rbd}\left(\CH\right),\nonumber
\end{equation}
where by ${rbd}\left(\CH\right)$ we denote the relative boundary of the convex hull $\CH$.
\end{theorem}

\begin{proof}
Since $\CH$ is a convex hull of $\V$, for all $\x\in\CH$
\begin{equation}
\x=\sum_{l=1}^L\lambda_l\x^l,\ \mbox{for some }\lambda_1,\dots,\lambda_L\geq0,\ \sum_{l=1}^L\lambda_l=1.\label{eqn:convex combination}
\end{equation}
For any $\th\in\mathbb{R}^d\setminus\left\{\bmath{0}\right\}$, let $\V_{\th}=\argmax_{\x\in\V}\inner{\th}{\x}$, and let $a_{\th}=\max_{\x\in\V}\inner{\th}{\x}$.
Then
$\forall\ \x\in\CH$, $\forall \x_{\theta}\in \V_{\th}$
\begin{align}
\inner{\th}{\x} &= \inner{\th}{\sum_{l=1}^L\lambda_l\x^l} = \sum_{l=1}^L\lambda_l\inner{\th}{\x^l} \leq \sum_{l=1}^L\lambda_l\inner{\th}{\x_{\th}}\nonumber\\
&= \inner{\th}{\x_{\th}}=a_{\th}.\nonumber
\end{align}
Then $\left\{x\in\mathbb{R}^d:\ \inner{\th}{\x}\leq a_{\th}\right\}$ is a supporting hyperspace, and $\H\left(a_{\th}\right)=\left\{x\in\mathbb{R}^d:\ \inner{\th}{\x}= a_{\th}\right\}$ is a supporting $d-1$ dimensional hyperplane for $\CH$.  Since $\CH$ is a full-dimensional convex polytope, $\CH\not\subset\H\left(a_{\th}\right)$, and therefore, $\H\left(a_{\th}\right)$ is a proper supporting hyperplane.  Thus $\mbox{rint}\left(\CH\right)\cap \H\left(a_{\th}\right)=\emptyset$ (e.g., \cite{Brondsted83}, Thm. 4.1).  Observing that $\mathcal{M}_{\th}\subset \H\left(a_{\th}\right)$ completes the proof.
\end{proof}
We now consider the formulation in \eqref{eqn:ergm weighted}, from which after taking the log-likelihood we observe that
\begin{align}
\argmax_{\x\in\feat}\ln P\left(\x\vert\th\right) &= \argmax_{\x\in\feat}\left\{\inner{\th}{\x}+\ln w\left(\x\right)\right\}.\label{eqn:max ergm weighted}
\end{align}

We extend $\feat$ to include the weight $w$.  Let\\$\bar{\feat}=\left\{\left(\x,\ln w\left(\x\right)\right):\x\in\feat\right\}$ be the extended set of features, with $\bar{\CH}$ being the resulting extended convex hull for $\bar{\feat}$ and $\bar{\V}$ the set of vertices for $\bar{\CH}$.  Then
\begin{align}
\inner{\th}{\x}+\ln w\left(\x\right) &= \inner{\left(\th,1\right)}{\left(\x,\ln w\left(\x\right)\right)}.\nonumber
\end{align}
By Theorem \ref{thm:max inner product}, only the points on the boundary $\mbox{rbd}\left(\bar{\CH}\right)$ can maximize \eqref{eqn:max ergm weighted}.  Thus to find the set of possible modes of \eqref{eqn:ergm weighted}, we can restrict our attention only to\\ $\mathcal{M}=\left\{\x:\left(\x,\ln w\left(\x\right)\right)\in\bar{\V}\right\}$.

Consider an illustration in Figure \ref{fig:fpmodes3D} for the case of all $8$-node graphs.  Only points on the boundary of the polytope $\bar{\CH}$ can {\em potentially} be modes of any ERGM specified on the feature set in Figure \ref{fig:fpmodes} (points on the boundary are denoted by solid black discs and solid red discs with black borders on both plots).  Figure \ref{fig:fpmodes} further reveals that only a small number of points on the relative boundary of $\bar{\CH}$ correspond to {\em actual observed} modes for equation \ref{eqn:ergm weighted} (ERGMs with parameters $\thxstar$ from (\ref{eqn:maxent})).\footnote{{\scriptsize The solid red discs in Figure \ref{fig:fpmodes} handles the case where modes of the maximum likelihood distribution are not unique.}}  There are two possible reasons for this occurrence.  One, some of the points in $\bar{\feat}$ correspond to $\bar{\th}=\left(\th,t\right)$ with $t\neq 1$ (where $t\in \mathbb{R}$, is a scalar augmenting the parameter vector), so not all of $\mathcal{M}$ may be the modes of \eqref{eqn:ergm weighted}.  Two, the MLE solution $\thxstar$ may lie outside of the cone \cite{Rinaldo2009} $$\left\{\th\in\mathbb{R}^d:\ \xstar\in\argmax_{\x\in\CH}\inner{\th}{\x}\right\}$$ of parameters for which $\xstar$ is the mode.\footnote{{\scriptsize Often, feature vectors corresponding to the empty and to the complete graphs have particularly large cones, with many $\thxstar$ values falling within these cones, where $\xstar$ is the feature maximizing (\ref{eqn:max ergm weighted}).  This explains why many ERGMs place large probability masses on these degenerate graphs.}}

The above observations also explain why considering a curved exponential model \cite{Hunter2006}
\begin{equation}
\begin{split}
P\left(\x\vert \th\right) &= \frac{1}{Z\left(\th\right)}w\left(\x\right)\exp\inner{g\left(\th\right)}{\x},\ \x\in\feat
\end{split}\nonumber
\end{equation}
does not rectify the issue of degeneracy as curving the space of parameters does not change the geometry of the feature space.\footnote{{\scriptsize We observed the same $ERGM$ behavior in separate experiments on other types of features spaces {\em e.g. Edge-vs-2Star, Edge-vs-GWED, Triangle-vs-GWESP} (GWED: geometrically weighted degree distribution, GWESP: geometrically weighted edgewise shared partners \cite{Hunter2008a}.)}}

Finally, we note that this type of degeneracy is not restricted to distributions over statistics of finite graphs.  Similar issues can in general arise with exponential family distributions on a bounded support owing to the fact that the exponential family distributions are designed to match the mean statistics and not to concentrate the probability mass around the mode.  Consider an illustration in Figure \ref{fig:degeneracyongrid} for a $20 \times 20$ grid of uniformly spaced points; fitted is the exponential family model $p\left(\bmath{x}\vert\th\right)\propto \exp\left<\th,\x\right>$. As evident from the plots, the estimated model exhibit the degeneracy issue described above.

\begin{figure*}[ht]
\centering
\begin{tabular}{cc}
\hspace{-0.8cm}\includegraphics[scale=0.8]{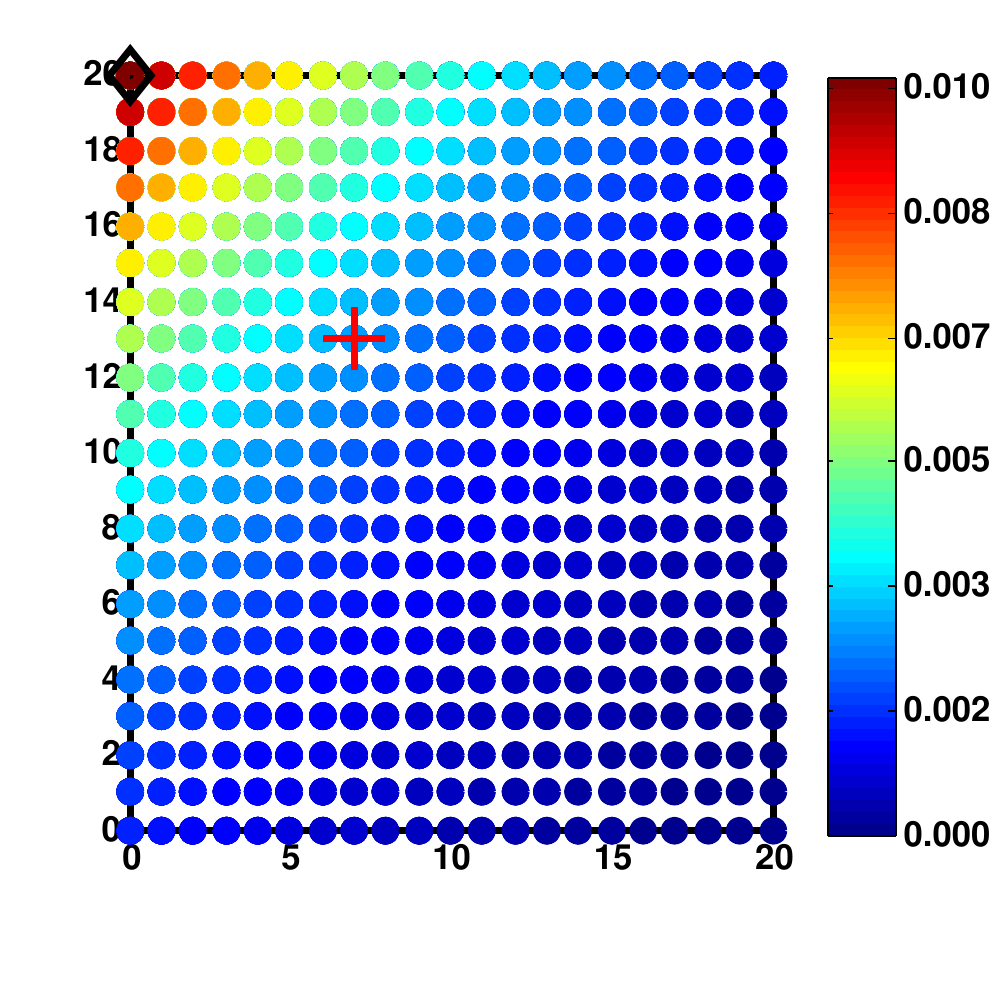} &\hspace{-0.6cm} \includegraphics[scale=0.8]{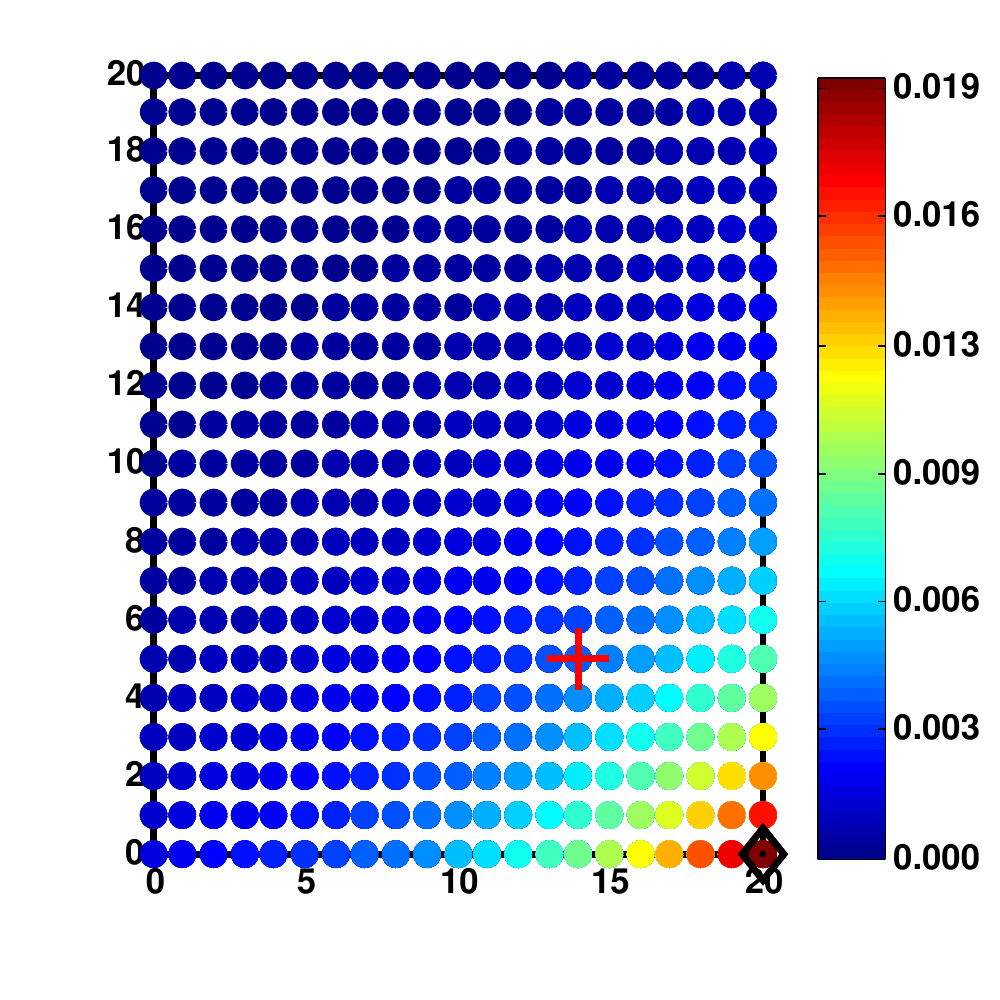}
\end{tabular}
\caption{An illustration of degeneracy of exponential models on non-graph data. Colorcoded is the pmf over 2D grid spaces for estimated MLEs {\em (left plot)} $\th_{MLE}=(-0.086,0.086)$ $\&$ {\em (right plot)} $\th_{MLE}=(0.120,-0.160)$.  "+" sign indicates the observed feature and its mean, while the "$\diamond$" shape indicates ERGM mode.}
\label{fig:degeneracyongrid}
\end{figure*}
\section{Exponential Locally Spherical Random Graph Model}\label{sec:elsrgm}
The main result of section \ref{sec:ergm} provides us with a very important insight, that is; type II degeneracy in $ERGMs$ is due to the bounded nature of discrete exponential models. The formulation of such models is sensitive to the geometry of the support space $\feat$. The geometry of the convex hull on $\feat$ defines which points are likely to have {\em most} or {\em all} of the probability mass placed on them. If $\x\not\in\mbox{rbd}(C)$ then a model computed for $\x$ will place very little probability mass on $\x$, while most of the mass is placed on some point {\em (i.e. mode)} ${\x}^{\star}\in\mbox{rbd}(C)$. This result suggests that mapping all points $\x\in\feat$ onto a surface that belongs to its relative boundary and defining a model in the new space would solve the degeneracy issue. In the following sections, we propose mapping all observed features onto a spherical surface $\mathbb{S}^{p-1}$ since every point on it
belongs to $\mbox{rbd}(C(\mathbb{S}^{p-1}))$. We then define a distribution and sampling techniques for graphs over the resulting feature space.

\subsection{Algorithm}

\begin{algorithm}[t]
\caption{Outline of {\tt ELSRGM} Procedure}
\label{fig:lsrgm}
\begin{algorithmic}[1]
\REQUIRE{Given an example graph $\Gstar\in\Gn$ and a feature vector function $\bmath{f}:\Gn\to\mathbb{R}^d$:}
\STATE Generate a neighborhood $\nb\subseteq\Gn$ by performing a random walk in $\Gn$ starting at $\Gstar$.
\STATE Compute the set of statistics $\feat\left(\nb\right)$.
\STATE Compute a mapping $\emb:\mathbb{R}^d\to\Sp$ to embed $\feat\left(\nb\right)$ onto $\Sp$, a surface of a sphere in $\mathbb{R}^p$.
\STATE Estimate the parameters for the von Mises-Fisher density $f$ over the space of hyper-spherical features $\Sp$.
\STATE Approximate $f$ with a density $\hat{f}$, a mixture of kernels centered around the spherical features corresponding to graphs in $\nb$ and recompute $\nb$ when new sample graphs are discovered.
\end{algorithmic}
\end{algorithm}

In this section we describe the new model for graph sampling, one that can be used to generate non-degenerate graphs similar to the given one.  Since the approach we are proposing is based on an exponential family model over the locally spherical embeddings of graph features, we will call it an Exponential Locally Spherical Random Graph Model, or \ELSRGM\ for short.  Our approach is executed in several steps as outlined in Algorithm \ref{fig:lsrgm}.

First, we sample the neighborhood $\nb\subseteq \Gn$ around the given example graph $\Gstar$, and then compute the set $\feat\left(\nb\right)=\left\{\fG:G\in\nb\right\}$ of feature vectors corresponding to the graphs in the neighborhood.  If the space of graphs ($\Gn$ or its subset) is small, then it may be possible to consider all graphs in the set.  Otherwise, the neighborhood of $\Gstar$ is sampled by a random walk which at each step considers graph one edge deletion/insertion away from the current graph.

The resulting feature set is then embedded in a $p-1$ dimensional unit hypersphere $\Sp\subset\mathbb{R}^p$ by a linear mapping $\emb:\mathbb{R}^d\to\Sp\subset \mathbb{R}^p$.  We are using the spherical embedding approach of \cite{Wilson2010}, Algorithm \ref{fig:SphericalEmbeddings}), in which the mapping is chosen to minimize the Frobenius matrix distance between the normalized dissimilarity matrix (in our case, matrix of Euclidean distances between feature vectors in $\feat\left(\nb\right)$) and the matrix of the Euclidean outer product (in $\mathbb{R}^p$) for the vectors in $\left\{\emb\left(\bmath{x}\right):\ \bmath{x}\in\feat\left(\nb\right)\right\}$.  We denote the set of resulting spherical features by $\SB=\left\{\emb\left(\fG\right):G\in\nb\right\}$.\footnote{{\scriptsize We refer to the spherically embedded graph feature vectors as {\em coordinates} of the graph.}}  One of the beneficial properties of such embedding is that it preserves neighborhood properties, i.e., transformed feature vectors close to each other are mapped to vectors which are also close to each other.  For our case, the embedding is {\em locally} spherical as we determine the mapping based only on a small subset of possible observed graphs, and determining the spherical coordinates for the rest by recomputing the embedding iteratively based on the uncovering of new candidate graphs for the neighborhood $\nb$. As a result, the distance preserving property may hold only for the graph neighborhood on which the transformation was estimated.

\begin{algorithm}[t]
\caption{Outline of Spherical Embedding}
\label{fig:SphericalEmbeddings}
\begin{algorithmic}[1]
\REQUIRE{Given dissimilarity matrix $D_{n\times n}$, with $n$ number of graphs.}
\STATE If the spherical point positions are given by $X_{i},\ i=1,\cdots,n$, then $\langle X_{i},X_{j}\rangle = r^{2}\cos\beta_{ij},$ with $\beta_{ij} = \frac{d_{ij}}{r}.$
\STATE If $X$ in unknown, compute for $X$ such that $XX^{T} = Z,$ where $Z_{ij}=r^{2}\cos\beta_{ij}$ and $d_{ij}\in D$. Find the radius of sphere as  $ r^{\star} = \arg\min_{r}\lambda_{1}\{Z(r)\}$. Where $\lambda_1$ is the smallest eigenvalue of $Z$.
\STATE Set $\hat{Z} = \frac{Z}{r^{\star}}$ and $ X^{\star}= \arg\min_{X,x^{T}x=1}\|XX^{T} - \hat{Z}\|$
\STATE Decompose $\hat{Z}$, $\hat{Z}=U\Lambda U^{T}.$ Set the embedding positional matrix to be $X^{\star}= U_{n\times k}\Lambda_{k\times k}^{1/2},$ where $k$ is chosen such that the elements of $U_{n\times k}$ corresponds to the largest $k$ eigenvalues of $\Lambda_{k\times k}$.
\end{algorithmic}
\end{algorithm}

The next step is to estimate a distribution over all possible values in $\Sp$ of the spherical features.  We propose to use von Mises-Fisher directional distribution (denoted in the rest of the paper as VMFD) with a pdf
\begin{equation}
f_{vmfd}(\bmath{y}|\bmath{\mu},\kappa) = \frac{\kappa^{\frac{p}{2}-1}}{(2\pi)^{\frac{p}{2}}I_{\frac{p}{2}-1}(\kappa)}\exp\left(\kappa\bmath{\mu}^T\bmath{y}\right),\ \bmath{y}\in\Sp
\label{eqn:vonmf}
\end{equation}
where $\bmath{\mu}$ is the location parameter, $\kappa\geq0$ is the concentration parameter, and $I_k$ denotes the modified Bessel function of the first kind and order $k$.  VMFD is a member of the exponential family; unlike a general exponential family distribution, it is symmetric with $\bmath{y}=\bmath{\mu}$ serving as both its mean and mode.  Parameter $\kappa$ determines how concentrated the density is around the mode. When $\kappa=0$, VMFD corresponds to the uniform distribution over the hyper-sphere $\Sp$, while as $\kappa\rightarrow\infty$, VMFD is concentrated at the point $\bmath{\mu}$.  See \cite{Mardia2000} for more details on VMFD.  As we wish the example graph $\Gstar$ to be the mode for the distribution over possible graphs, we set the mode of VMFD to $\hat{\bmath{\mu}}=\bmath{e}\left(\bmath{f}\left(\Gstar\right)\right)$. Using MLE, one can compute for $\hat{\kappa}$ from $\SB$ with already set location parameter $\hat{\bmath{\mu}} $i.e. $\hat{\kappa}_{mle} = \frac{p-1}{2(1-\hat{\bmath{\mu}}^{T}\sum_{G\in\SB}\bmath{e}\left(\bmath{f}\left(G\right)\right)/{|\SB|})}$. Given only a single example graph, the concentration parameter is undefined and also, $\hat{\kappa}_{mle}$ can be observed to be a function of the random-walk-based sampling procedure, which is arbitrary. An alternative is to treat $\hat{\kappa}$ as a user-defined parameter controlling how concentrated the region around the example graph is. As a result, one obtains a distribution with density
\begin{equation}
f\left(\bmath{y}\right)\equiv f_{vmfd}\left(\bmath{y}\vert\hat{\bmath{\mu}},\hat{\kappa}\right)\label{eqn:vmfoversphere}
\end{equation}
over possible spherical feature vectors $\bmath{y}\in\Sp$.

We note that one can employ a different family of distributions other than VMDF for the task of modeling spherical data.  Other candidates provide more degrees of freedom, but also more difficult parameter estimation methods \cite{Mardia2000}.  We leave the investigation of this aspect of our algorithm for future work.

\subsection{Sampling under ELSRGM}
Using the ideas and steps from the previous subsection, we estimate a probability density function over a unit sphere on the domain of possible spherical feature values for graphs.  Before proceeding further, however, we need to resolve several issues.  First, we are interested in a discrete probability distribution over a very large but {\em finite} set of possible features $\SB$.  How would one obtain such a distribution $P$ from the VMFD density?  Second, what portion of the mass of $f$ is to be associated with each of the coordinates for possible graphs $G\in\Gn$ (or $\bmath{y}\in\SB$)?  Third, it is infeasible to enumerate all of the possible graphs in $\Gn$, and it may be infeasible to consider all possible spherical feature vectors.  How can one perform sampling in $\Gn$ so that the resulting spherical features are distributed according to $P$? The following subsections will outline our approach to resolving these issues.

\subsubsection{Density as an Approximation to the Smoothed Distribution}
As a first step, consider a setting where all possible spherical feature vectors $\SB$ can be enumerated for a fixed number of nodes $n$.  We propose to approximate VMFD using a mixture of kernels with one mixture component for each graph's spherical coordinates.  (Alternatively, one can consider one mixture component for each possible spherical feature vector in $\SB$.)  Intuitively, we consider the density $f$ to be a baseline approximation to the distribution obtained by smoothing a discrete distribution over spherical feature values $\SB$.  More formally, for a set of graphs $\mathcal{G}$, we will index its elements $G_1,\dots,G_K$ with $K=\left|\mathcal{G}\right|$, and denote the spherical feature vector for graph $G_k$ by $\bmath{y}_k=\emb\left(\bmath{f}(G_{k})\right)\in\SB$ .

Assume $f\left(\bmath{y}\right)$ is the estimated density in Equation (\ref{eqn:vmfoversphere}).  Let
\begin{equation}
\hat{f}\left(\bmath{y}\vert\bmath{\pi}\right)=\sum_{k=1}^K\pi_kh_k\left(\bmath{y}\right)\label{eqn:kernelmixture}
\end{equation}
with $h_k\left(\bmath{y}\right)\equiv f_{vmfd}\left(\bmath{y}\vert\bmath{y}_k,\kappa_{h_k}\right)$, where $\bmath{y}_k$ is the location for kernel $h_k$, and all of the kernels share a user defined concentration (bandwidth) $\kappa_{h_k}$ chosen to assign more weight to points closer to $\bmath{y}_k$.  The parameters $\bmath{\pi}=\left\{\pi_k:k=1,\dots,K\right\}$, $\pi_k\geq 0$, $\sum_{k=1}^K\pi_k=1$ are probabilities associated with each $\bmath{y}_k$.

The estimation of the $\pi$'s is carried out by optimizing with respect to the Kullback-Leibler (KL) information:
\begin{eqnarray}
J\left(\bmath{\pi}\right) = KL\left(\hat{f}\left(\cdot\vert\bmath{\pi}\right)\parallel f\right) - \lambda\left(\sum_{k=1}\pi_k-1\right)
\label{eqn:piobjective}
\end{eqnarray}
where $\lambda$ is the Lagrange multiplier.  It can be shown that the objective function $J\left(\bmath{\pi}\right)$ in \eqref{eqn:piobjective} is convex, and it can be minimized efficiently using standard convex optimization techniques.

\subsubsection{Metropolis-Hastings Algorithm}
The goal is to draw samples from $\Gn$ according to the probability mass $\{\pi_1,\dots,\pi_K\}$, where  $K=|\Gn|$. If the number of nodes $n$ is small ($\leq12$), random graphs can be enumerated using {\tt nauty} (\cite{McKay1981}) and it is possible to identify all possible feature values to compute probability masses $\pi_1,\dots,\pi_K$ associated with each graph.  Graphs can then be sampled directly according to the resulting multinomial distribution. When all possible graphs of $n$-nodes (for $n\leq 10$) are observed, an alternative is to make use of the Metropolis-Hastings approach as outlined in Algorithm \ref{alg:MCMC}.  However, in practice, the number of nodes is usually too large to explicitly consider all possible graphs, and the initial neighborhood $\nb$ would include only a small portion of all graphs.  We propose an approach that will allow us to draw samples from the distribution over graphs including the graphs in $\Gn\setminus\nb$.

For better understanding, we first propose a graph generation approach assuming $n$ is small, and all graphs $G_{1},\dots,G_{K}$ can be enumerated. We will draw a point in the embedding space, and then employ Markov Chain Monte Carlo (MCMC) approach to draw a graph ``corresponding'' to this point by constructing a Markov chain in the space $\Gn$ of graphs.   The pseudocode is presented in Algorithm \ref{alg:MCMC}.  If we could draw samples directly from $\hat{f}$, this procedure would be equivalent to drawing a vector $\bmath{y}^\star\sim \hat{f}$, and then trying to identify out which of the $K$ components was used to generate $\bmath{y}^\star$ by using MCMC in the posterior over $k=\left\{1,\dots,K\right\}$.  The only approximation employed in this case is that instead of drawing $\bmath{y}^\star\sim\hat{f}$, we are drawing $\bmath{y}^\star\sim f$, but according to the KL-divergence measure (Figure \ref{fig:g8snythetic-convergence-8762}), $\hat{f}$ and $f$ are very close.

\begin{algorithm}[t]
\caption{MCMC Algorithm for Sampling Graphs}
\label{alg:MCMC}
\begin{algorithmic}[1]
\REQUIRE{$f$ is given.}
\STATE Draw a sample $\bmath{y}^\star\sim f$. Set $t=0$, and perturb the example graph $G\in\Gn$ to generate a random graph $G^0\in\Gn$ to initialize the chain.
\REPEAT
    \STATE $t\leftarrow t+1$.
    \STATE Sample a proposal graph $G^\prime$ from the proposal distribution at time $t$, $Q(G^\prime|G^{t-1})$.
	\STATE Compute ratio, $ r = \frac{P\left(G^\prime|\bmath{y}^\star\right)Q\left(G^{t-1}|G^\prime\right)}{P\left(G^{t-1}|\bmath{y}^{\star}\right)Q\left(G^\prime|G^{t-1}\right)} $\\
where $ P(G=G_{k}|\bmath{y}) \propto \pi_kh_k\left(\bmath{y}\right).$
    \STATE $ G^{t} = \left\{
			     \begin{array}{lr}
			       G^\prime & \ with\  prob\ min(r,1)\\
			       G^{t-1} & o.w.
			     \end{array}
			   \right.$
\UNTIL{convergence of the chain}
\end{algorithmic}
\end{algorithm}

To sample from the von Mises-Fisher distribution $f$ we follow the approach outlined in \cite{Wood1994}. For proposal distribution $Q\left(\cdot\vert G^{t-1}\right)$ we consider a uniform distribution over graphs one edge insertion/deletion away from $G^{t-1}$.  In our experiments on small-and large-sized graphs, with the number of iterations set to $T=1000$, the algorithm is observed to converge and does produce graphs with topologies that resemble the observed graph.

For larger graphs, however, we cannot compute explicitly $\pi_k$ for all $K$ graphs.  This situation can be thought of as similar to the case of {\em countably infinite} number of objects in which a Dirichlet process can be employed to assign some weight to graphs that are not yet observed.  It is as if, $f$ is approximated by a countably infinite mixture $\hat{f}\left(\bmath{y}\right)=\sum_{k=1}^\infty\pi_kh_k\left(\bmath{y}\right)$, i.e., assuming that the number of graphs is countably infinite instead of just very large.  In this case, we are employing Dirichlet process mixture model $DP\left(\alpha,G_0\right)$, where $\alpha$ is a concentration parameter and $G_0$ is a uniform distribution $u_p\left(\bmath{\mu}\right)\propto 1,\ \bmath{\mu}\in\Sp$ for the kernel location since assume that each yet unobserved set of features is equally likely.

Assuming that $K$ different graphs have been observed, equation \ref{eqn:kernelmixture} then becomes
\begin{eqnarray*}
\hat{f}\left(\bmath{y}\vert\bmath{\pi}\right) = \sum_{k=1}^K\pi_kh_k\left(\bmath{y}\right)+\frac{\alpha}{K+\alpha}u_p\left(\bmath{y}\right)\label{eqn:kerneldensitydp}
\end{eqnarray*}

Algorithm \ref{alg:MCMCDP} details the pseudocode for sampling of small to large graphs on a spherical space.  For basic information on Dirichlet process mixtures and their inference see \cite{Neal00}.

\begin{algorithm}[t]
\caption{MCMC-DP Algorithm for Sampling Graphs}
\label{alg:MCMCDP}
\begin{algorithmic}[1]
\REQUIRE{$f$ is given. Graphs $G_{1},\dots,G_{L}$ initially observed.
Set $K=L$.  Estimate $\bmath{\pi}$ by minimizing \eqref{eqn:piobjective}.}
\STATE Draw a sample $\bmath{y}^\star\sim f$. Set $t=0$, and perturb\footnote{{\scriptsize perform edge random edge insertion/deletion atmost twice}} the example graph $G\in\Gn$ to generate a random graph $G^0\in\Gn$ to initialize the chain.
\REPEAT
\STATE $t\leftarrow t+1$.
	\STATE Sample a proposal graph $G^\prime$ from the proposal distribution at time $t$, $Q(G^\prime|G^{t-1})$.
	\STATE Compute ratio, $r= \frac{P(G^\prime\vert\bmath{y}^\star)Q(G^{t-1}|G^\prime)}{P(G^{t-1}|\bmath{y}^\star)Q(G^\prime|G^{t-1})}$
\STATE where $$P(G=G_{k}|\bmath{y})\propto
\begin{cases}
\pi_kh_k\left(\bmath{y}\right)&k=1,\dots,K,\\\frac{\alpha}{K+\alpha}u_p\left(\bmath{y}\right)&G=new.
\end{cases}$$
    \STATE Set $ G^{t} = \left\{
			     \begin{array}{lr}
			       G^\prime & \ with\  prob\ min(r,1)\\
			       G^{t-1} & o.w.
			     \end{array}
			   \right.$
    \STATE If $G^t$ is new, then set $K=K+1$, $G_{K}=G^t$.
      \begin{itemize}
        \item Compute $\bmath{f}(G_{K})$, set $\nb = \nb \cup \bmath{f}(G_{K})$.
        \item Re-compute $\SB$ ({\small to generalize to new samples}) and re-estimate $\pi_1,\dots,\pi_K$.
      \end{itemize}
\UNTIL{convergence of the chain}
\end{algorithmic}
\end{algorithm}
\section{Experimental Evaluation}\label{sec:experiments}
We adopt the common goodness of fit measures in network generation studies to investigate how well our model fits the observed data (i.e. by comparing the observed statistics with a range of the same statistics obtained from simulating many networks using the fitted model) \cite{Hunter2008a}. If the observed network is not typical of the simulated networks for a particular measure then the model is either degenerate or simply a misfit. We first consider the {\em degree distribution}, which is defined as the statistics: $D_{0},D_{1},\cdots,D_{n-1}$, with each $D_{i}$ representing the number of nodes with $i$ edges connected to them, divided by $n$. Secondly, we compute the {\em edgewise shared partner distribution}, which is defined as the statistics: $EP_{0},EP_{1},\cdots,EP_{n-2}$, with each $EP_{i}$ representing the number of edges in the graph between two nodes that share exactly $i$ neighbors in common, divided by the total number of edges. Thirdly, we compute the {\em triad census distribution}, which is the  proportion of $3-node$ sets having $0,1,2,$ or $3$ edges among them. Lastly, we compute the {\em minimum geodesic distance distribution}; which is the proportion of pairs of nodes whose shortest connecting path is of length $k$, for $k = 1,2,\cdots$.  Also, pairs of nodes that are not connected are classified as $k =\infty$.

In all our experiments, we define each feature vector to be, $$\fG=\left(f_{edge},f_{2\star},\cdots,f_{(11)\star},f_{\triangle}\right)^{T}$$ These features have been observed to be among the set of subgraph patterns that capture social interactions and network formation in real processes \cite{FrankStrauss}.

\subsection{Small graphs}
We first consider the experiments for $\mathcal{G}_n$ with $n=8$ nodes. Figure \ref{fig:test8graphs} (left) shows a two component $8$-node synthetic graph, $G_{test_1}$, while Figure \ref{fig:g8snythetic-convergence-8762} displays the $KL$-divergence $KL\left(\hat{f}\parallel f\right)$ computed by first observing $300$ $8$-node graphs (graph-edit distance $\leq3$ of neighborhood around $G_{test_1}$), and allowing the {\tt MCMC-DP} (Algorithm \ref{alg:MCMCDP}) to discover and fill-in the neighborhood set $\nb$, with the baseline $VMFD$ centered at the coordinates of $G_{test_1}$. $G_{test_1}$ was chosen to test our hypothesis for the extended feature space, that is if $\xstar\in\mbox{rint}(\bar{C})$ $\mbox{e.g.}\ \xstar=\bmath{f}\left(G_{test1}\right)\in\mbox{rint}\left(\bar{\CH}\right),$ then the corresponding
$ERGM$ will be degenerate in relation to the analysis of Theorem \ref{thm:max inner product}. $G_{test_2}$ is used to test our second hypothesis,
i.e. if $\xstar\in\mbox{rbd}(\bar{C})$ then the ERGM specified by $\xstar$ is most likely to generate realistic graphs. In Figure
\ref{fig:test8graphs} (right) we show an $8$-node synthetic graph $G_{test_2}$ whose extended feature vector
($\xstar=\bmath{f}\left(G_{test2}\right)\in\mbox{rbd}\left(\bar{\CH}\right)$).

For the experimental set-up, the bandwidth for each kernel $h_{k}$  is treated as a user defined
parameter and is set to $\kappa_{h_k}=400$, while the DP prior parameter is set to $\alpha=0.5$. Parameters for the $VMFD$ are set as discussed in
section \ref{sec:elsrgm}, with $\kappa$ fixed at $140$ to control the concentration of $f_{VMFD}$. We compute the parameters for the ERGM model $P$
and simulations using the {\sf statnet} package \cite{Handcock2008}.

Figure \ref{fig:ergm-elsrgm-g8sampledStats1}, depicts $100$ simulated summary results for
$\bmath{f}(G_{test_1})\in\mbox{rint}\left(\bar{\CH}\right)$.  The result obtained from our \ELSRGM\  (in unshaded blue box plots), displays no sign of degeneracy while the $ERGM$ result shows summary statistics that appear to relate to those of empty and complete graphs- a sign of
degeneracy. Figure \ref{fig:ergm-elsrgm-g8sampledStats2}, shows summary statistics of simulated networks obtained for the $ERGM$ and the \ELSRGM\ models when $\bmath{f}(G_{test_2})\in\mbox{rbd}\left(\bar{\CH}\right)$. The results show no sign of degeneracy from both the specified $ERGM$ and the proposed \ELSRGM\ . This confirms our second hypothesis, that is having extended feature vectors lie on the relative boundary of the convex hull enables the generation of realistic graphs from exponential family models.

\begin{figure*}[t]
\begin{center}
\includegraphics[width=2in]{g8-8762.pdf}\hspace*{0.5cm}
\includegraphics[width=2in]{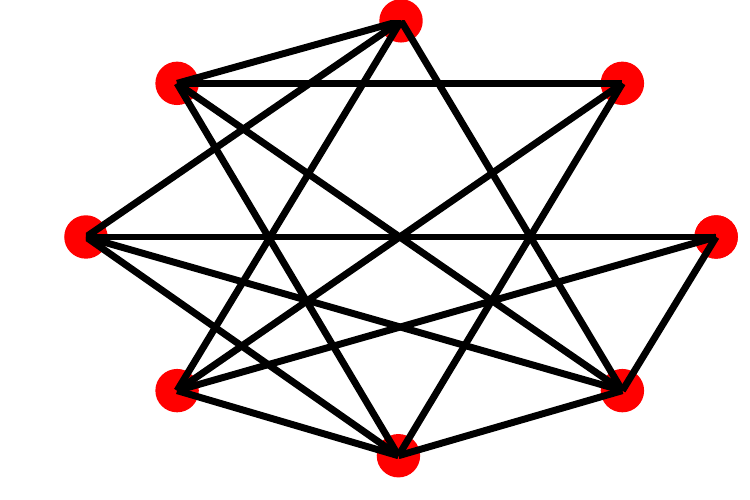}
\caption{Two 8-node synthetic test graphs.  {\bf Left:} 2-component graph $G_{test1}$ falling inside the relative interior of the convex hull $\bar{\CH}$ of extended features.  {\bf Right:} $G_{test2}$ falling on the relative boundary of the convex hull $\bar{\CH}$ for the extended feature space.}\label{fig:test8graphs}
\vspace*{0.5cm}
\includegraphics[scale=0.85]{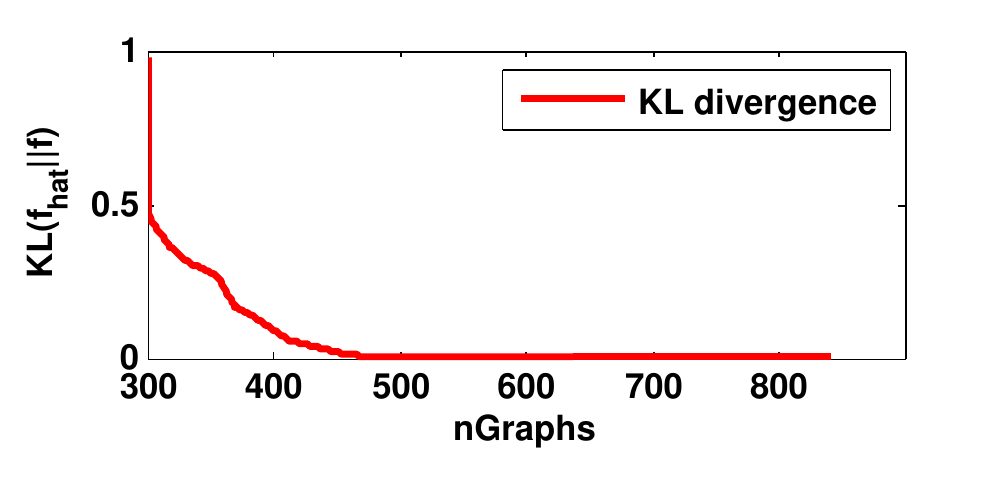}
\caption{Sample $KL$-divergence $KL\left(\hat{f}\parallel f\right)$ for \ELSRGM\ as more graphs get uncovered.  The initial neighborhood $\nb\subset\Gn (n=8)$ is built around $G_{test1}$.}
\label{fig:g8snythetic-convergence-8762}
\end{center}
\end{figure*}

\begin{figure}[htb]
  \begin{minipage}[t]{0.49\linewidth}\centering
    \includegraphics[width=2in]{g8-8762.pdf}
    \centerline{(a)}
  \end{minipage}\hfill
  \begin{minipage}[t]{0.49\linewidth}\centering
  \includegraphics[width=2.5in]{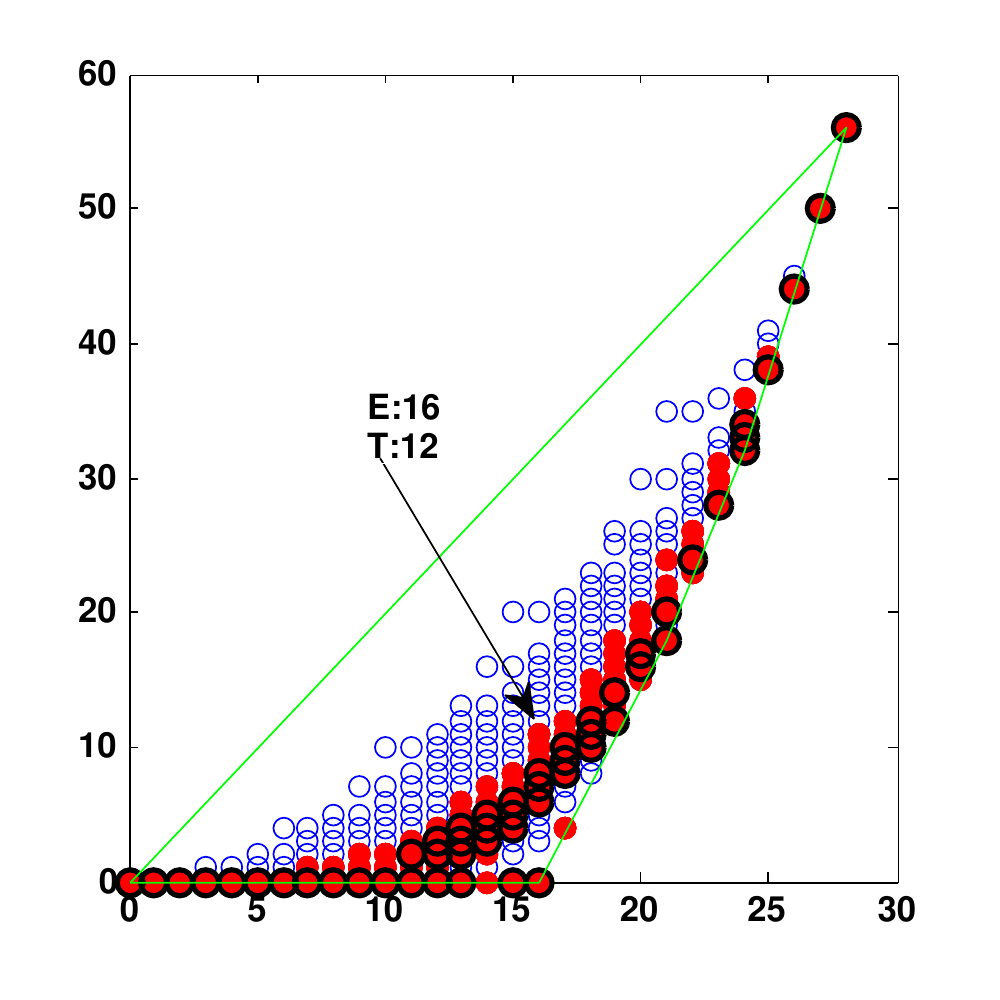}
    \centerline{(b)}
  \end{minipage}
  \bigskip
  \begin{minipage}[t]{0.49\linewidth}\centering
\begin{tabular}{cccc}
\hspace{-1cm}\includegraphics[scale=0.5]{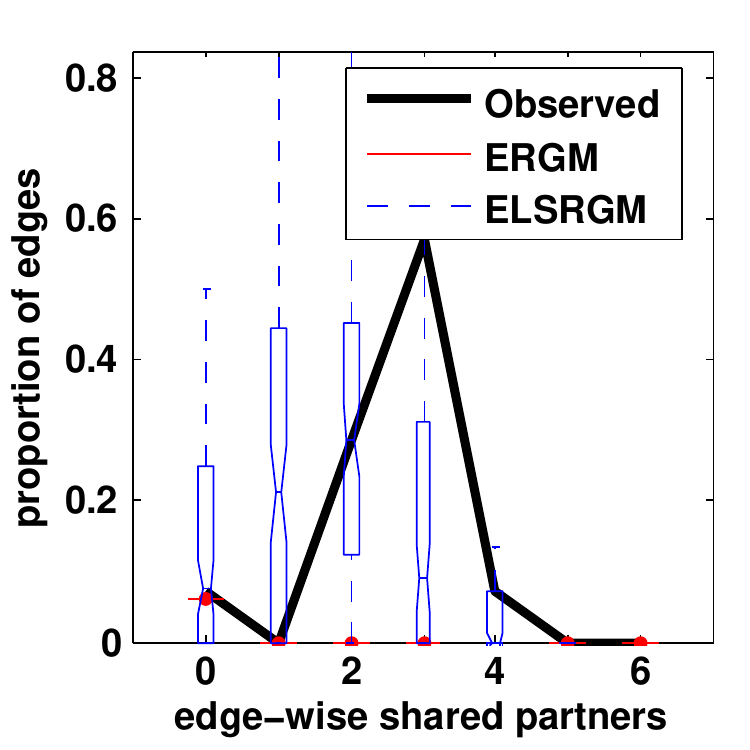} & \includegraphics[scale=0.5]{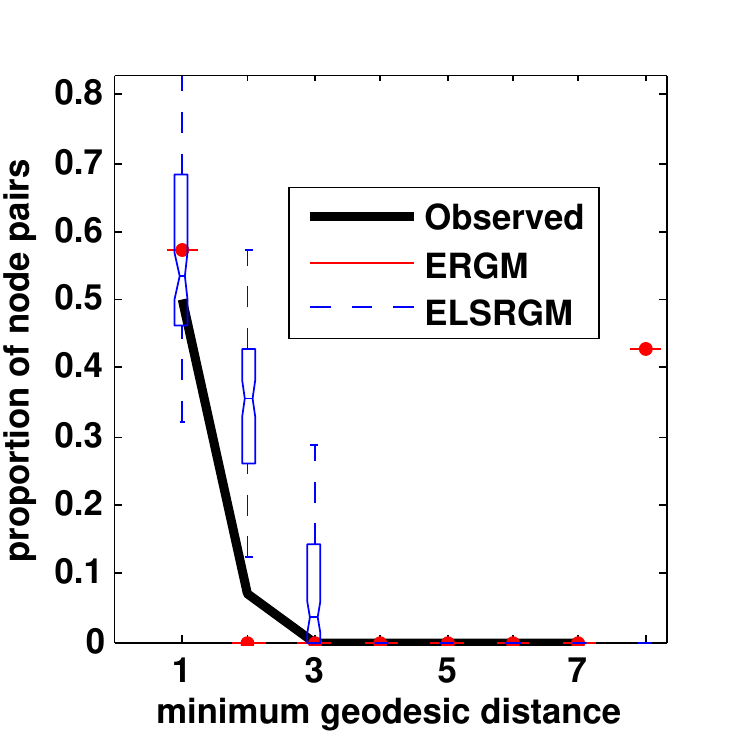} \includegraphics[scale=0.5]{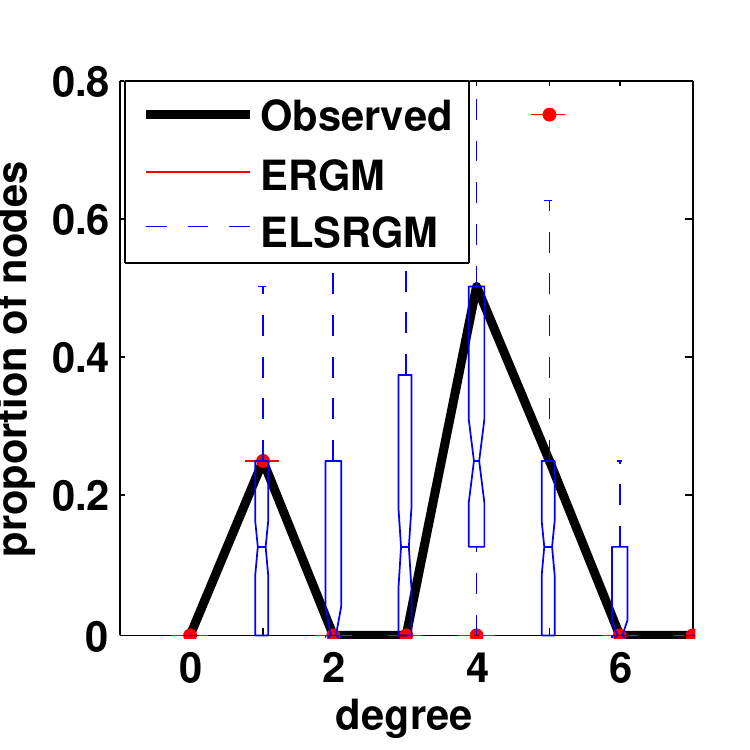}&\includegraphics[scale=0.5]{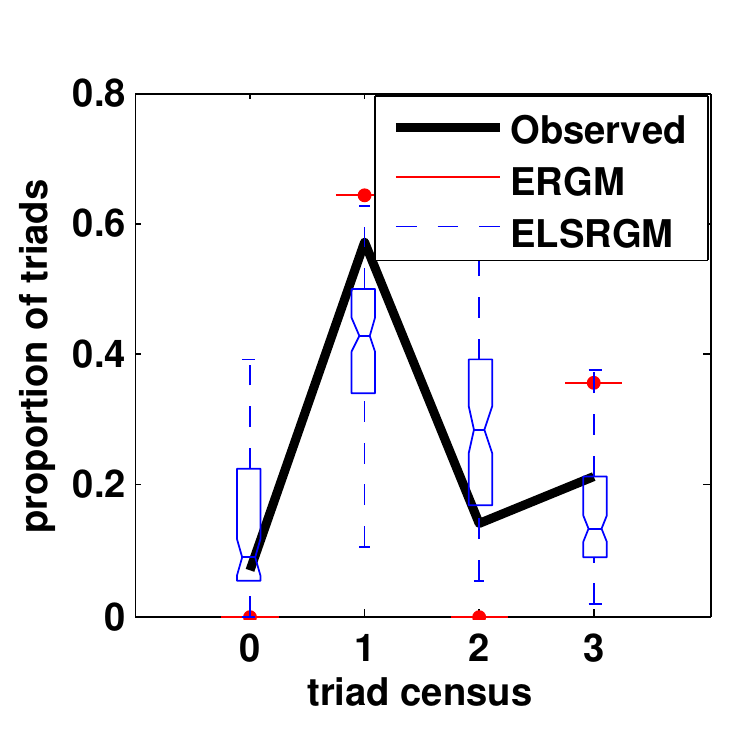}
\end{tabular}
    \medskip
    \centerline{(c)}
  \end{minipage}\hfill
  \caption{(a) Synthetic graph $G_{test_1}$; (b) Corresponding feature-pair: $\bmath{f}(G_{test_1})\in\mbox{rint}\left(\bar{\CH}\right)$; (c) Simulations from models given $G_{test_1}$. The observed statistics are indicated by the solid lines; the box plots include the median and interquartile range of simulated networks. $ERGMs$ show low variance and all $100$ samples seem to be placed on the same network- sign of degeneracy.}
  \label{fig:ergm-elsrgm-g8sampledStats1}
\end{figure}

\begin{figure}[htb]
  \begin{minipage}[t]{0.49\linewidth}\centering
    \includegraphics[width=2in]{g8-11000.pdf}
    \centerline{(a)}
  \end{minipage}\hfill
  \begin{minipage}[t]{0.49\linewidth}\centering
  \includegraphics[width=2.5in]{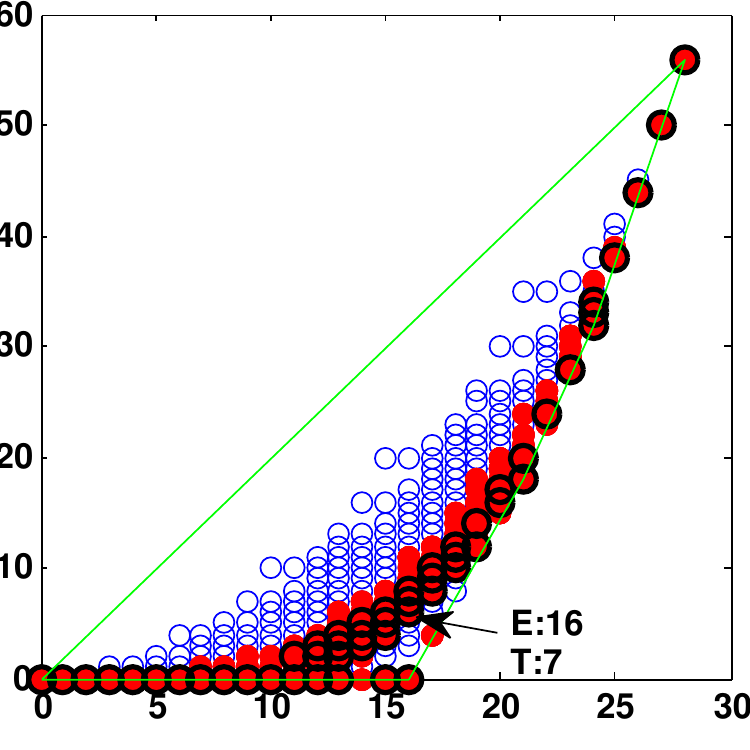}
    \centerline{(b)}
  \end{minipage}
  \bigskip
  \begin{minipage}[t]{0.49\linewidth}\centering
\begin{tabular}{cccc}
\hspace{-1cm}\includegraphics[scale=0.5]{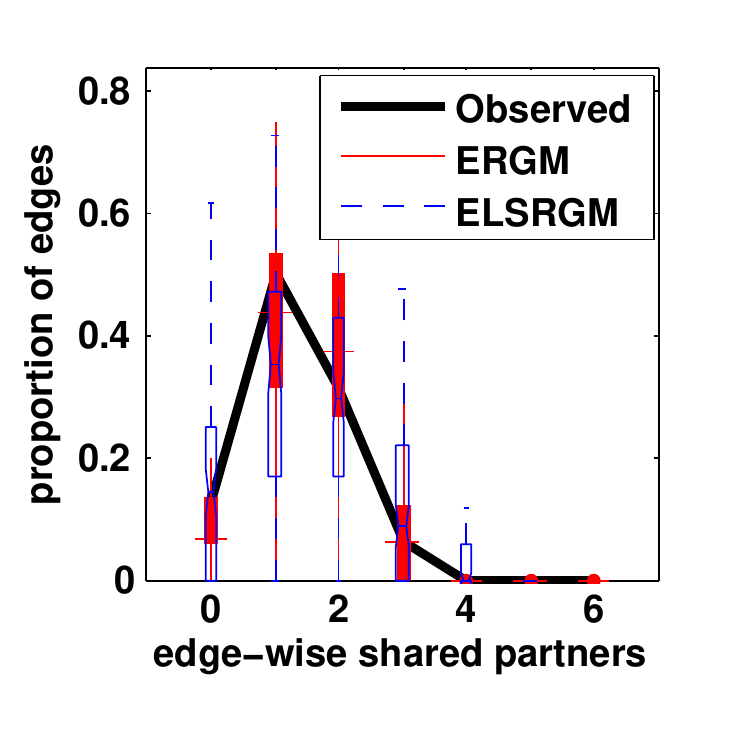} & \includegraphics[scale=0.5]{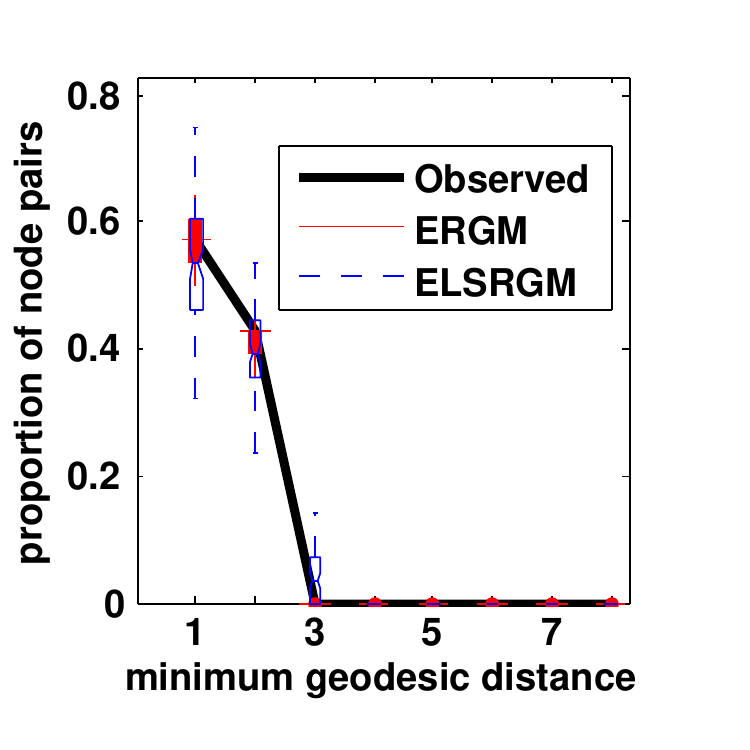}
\includegraphics[scale=0.5]{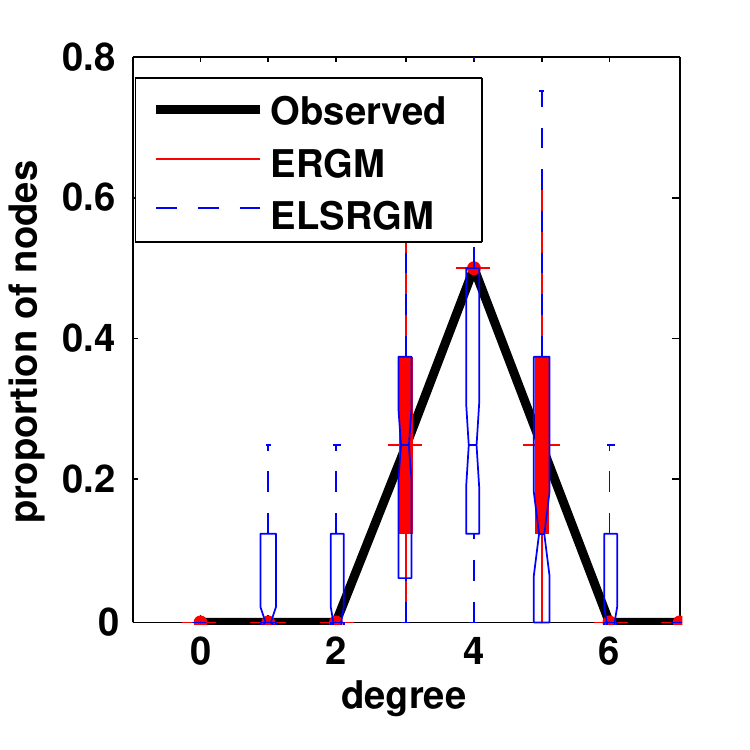}&\includegraphics[scale=0.5]{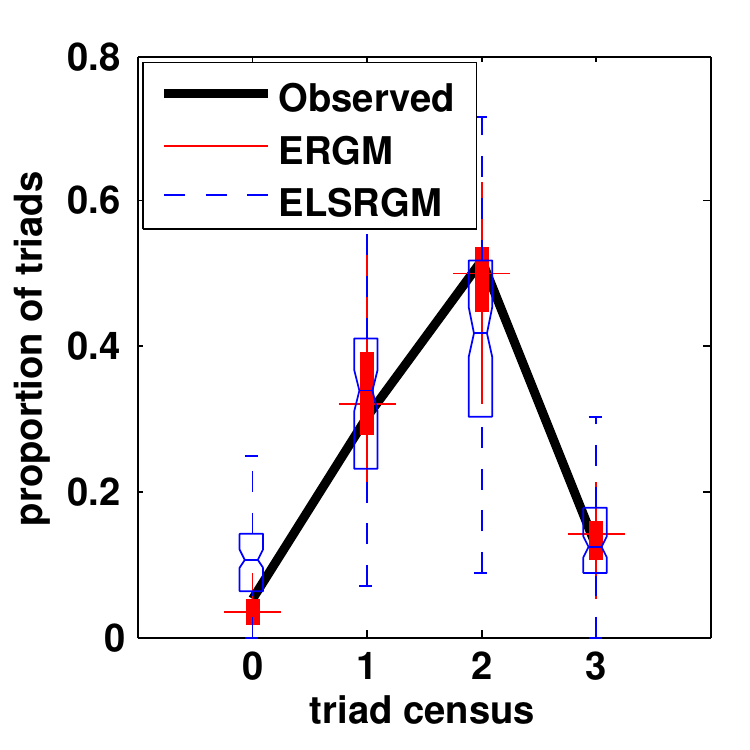}
\end{tabular}
    \medskip
    \centerline{(c)}
  \end{minipage}\hfill
  \caption{(a) Synthetic graph $G_{test_2}$; (b) Corresponding feature-pair: $\bmath{f}(G_{test_2})\in\mbox{rbd}\left(\bar{\CH}\right)$; (c) Simulations from models given $G_{test_2}$. The observed statistics are indicated by the solid lines; the box plots include the median and interquartile range of simulated networks. Both models are non-degenerate.}
  \label{fig:ergm-elsrgm-g8sampledStats2}
\end{figure}


\subsection{Larger Graphs}
For larger data sets, we first consider a undirected {\tt Dolphin} social network with $62$ nodes \cite{DuBois2008,Lusseau2003}. We apply the {\tt MCMCDP} approach as outlined in Algorithm \ref{alg:MCMCDP}, by sampling graphs via a random walk for $1000$ steps starting from a random network $G^0\in\nb$.  Figure \ref{fig:ergm-elsrgm-dolphin} summarizes the results of $100$ simulations for the {\tt Dolphin} network from the two approaches.  The proposed \ELSRGM\  shows a relatively better performance of capturing the distribution of statistics of the {\tt Dolphin} network.  It is appears that ERGM model specified by simple statistics is incapable of generating the distribution of statistics that resembles those of the observed network. The lack of fit by an ERGM model in the degree distribution, geodesic distance, and shared partners distribution indicates the presence of degeneracy.

\begin{figure*}[!hp]
\centering
\begin{tabular}{cccc}
\hspace{-1cm}\includegraphics[scale=0.5]{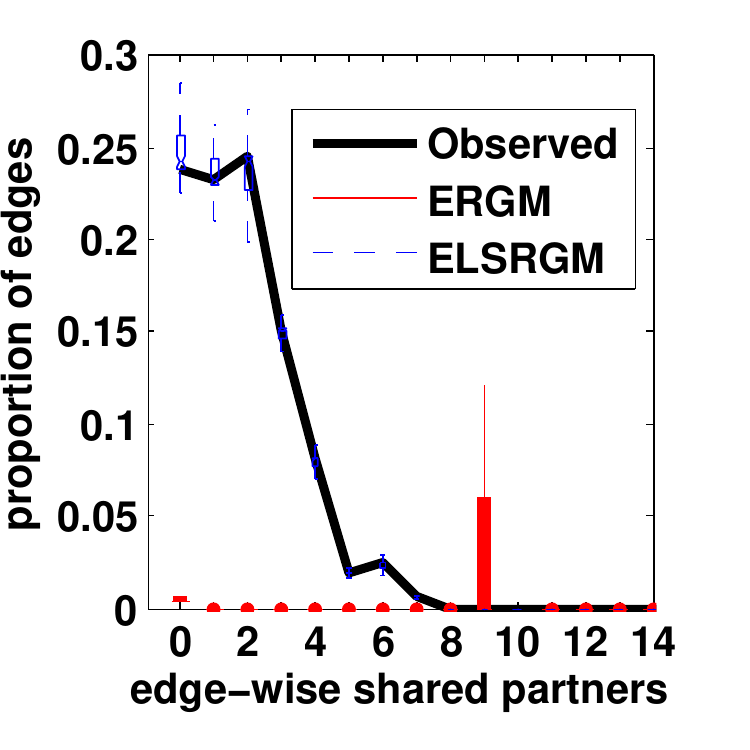} & \includegraphics[scale=0.5]{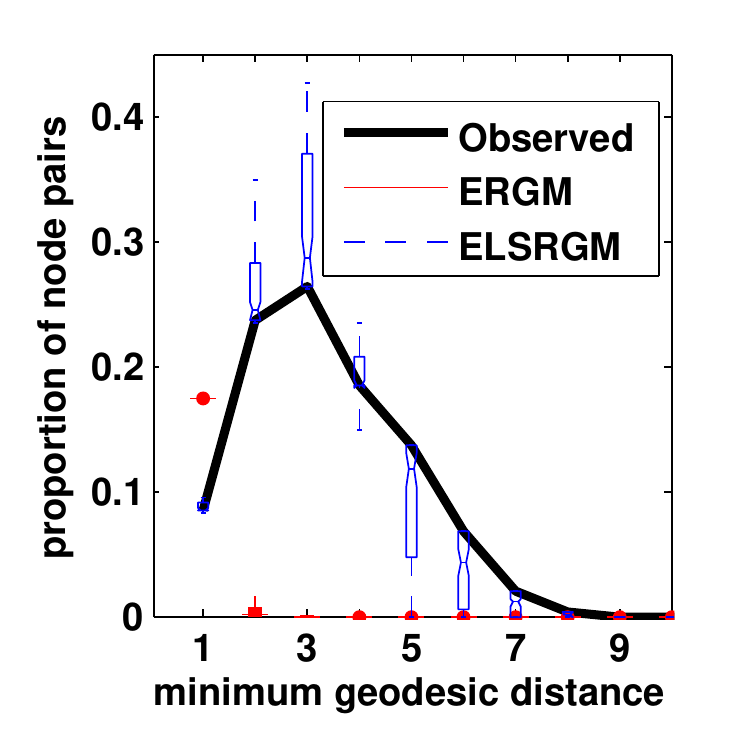}&\includegraphics[scale=0.5]{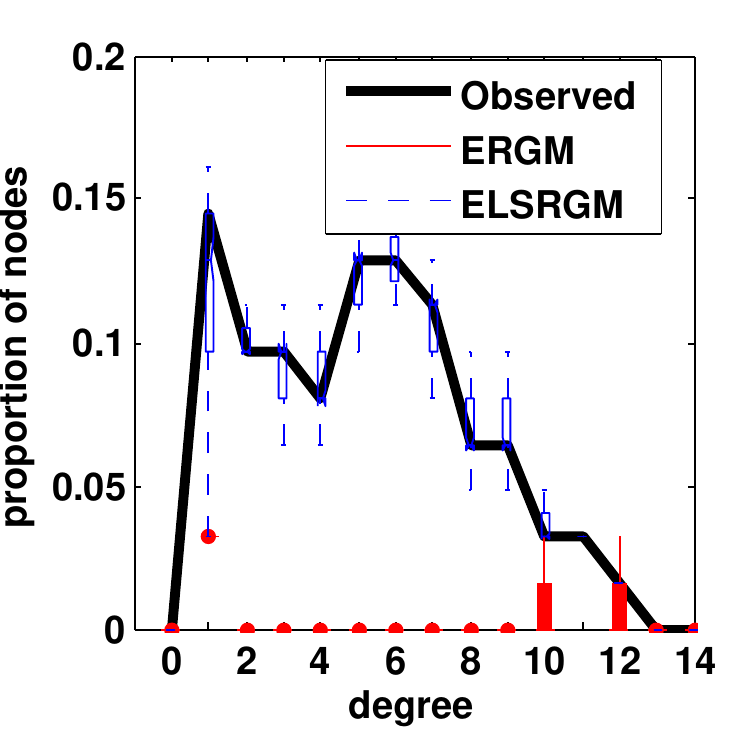} &\includegraphics[scale=0.5]{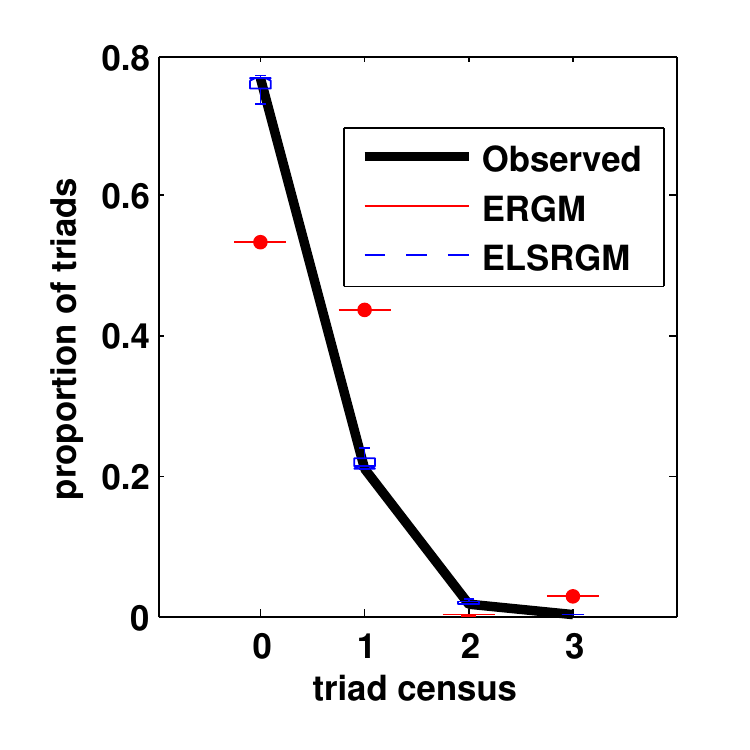}
\end{tabular}
\caption{{\tt Dolphin} $62$-node network summaries. The observed statistics are indicated by the solid lines; the box plots include the median and interquartile range of simulated networks. $ERGM$ plots shows signs of degeneracy effect.}
\label{fig:ergm-elsrgm-dolphin}
\end{figure*}

We next evaluate \ELSRGM\ on $205$-node {\tt Faux-Mesa-High} social network \cite{Handcock2008,Resnick1997}.
We again apply the {\tt MCMC-DP} algorithm  by sampling graphs via a random walk for $1000$ steps starting from a random network $G^0\in\nb$.
Figure \ref{fig:ergm-elsrgm-faux} depicts the results of $100$ simulated networks from the specified \ELSRGM\ and $ERGM$ models. Again, we observe that \ELSRGM\ generates distributions of statistics that resemble those of the observed network while ERGMs shows signs of misspecification.


\begin{figure*}[!hp]
\centering
\begin{tabular}{cccc}
\hspace{-1cm}\includegraphics[scale=0.5]{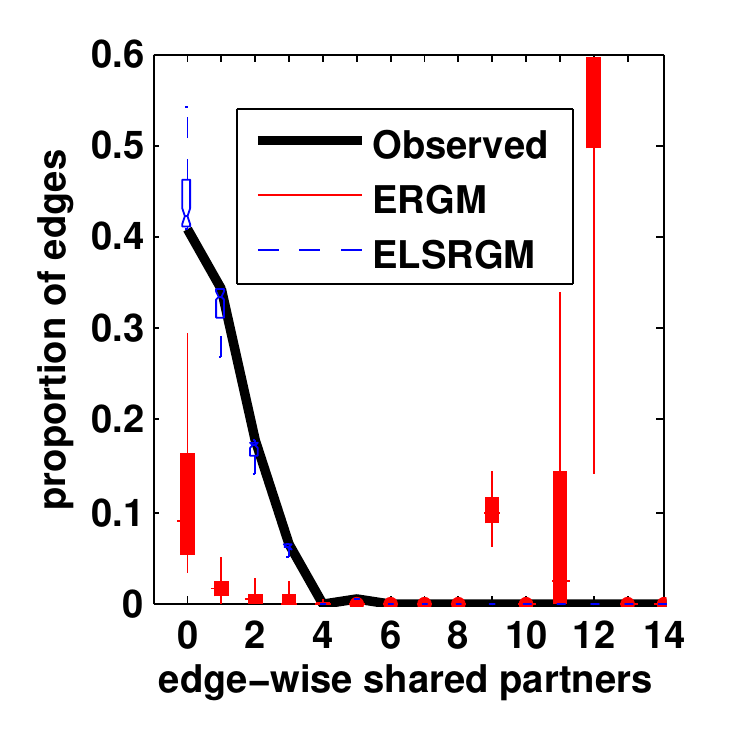} & \includegraphics[scale=0.5]{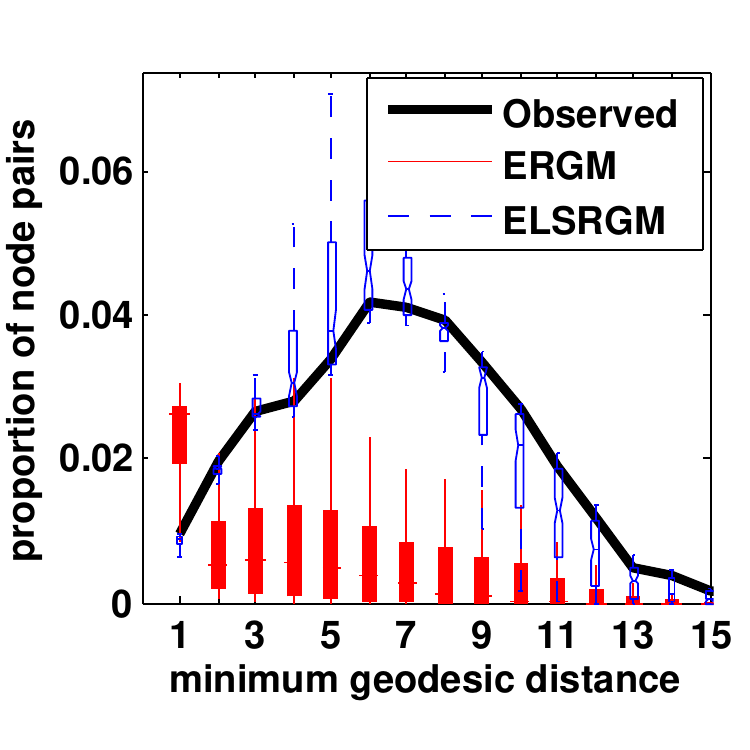} & \includegraphics[scale=0.5]{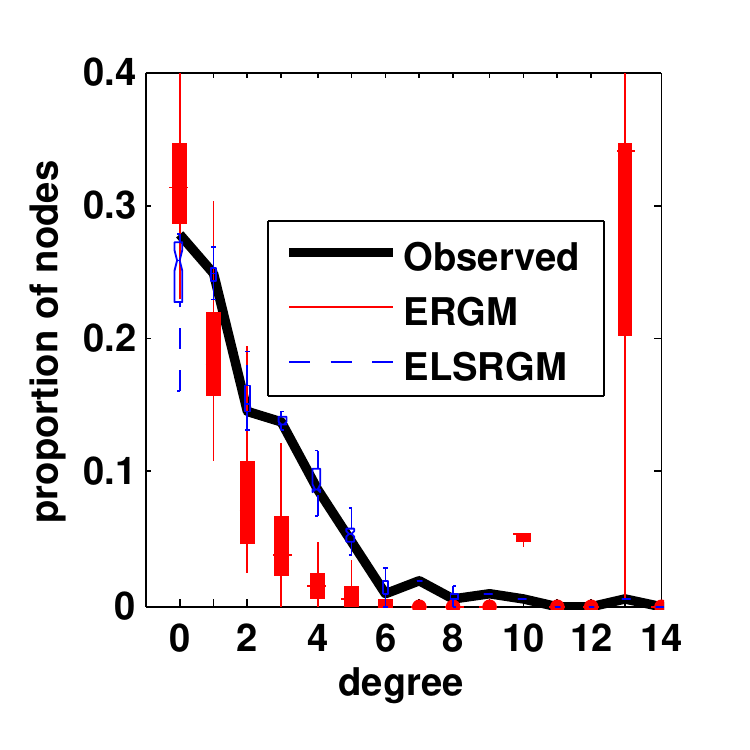}&\includegraphics[scale=0.5]{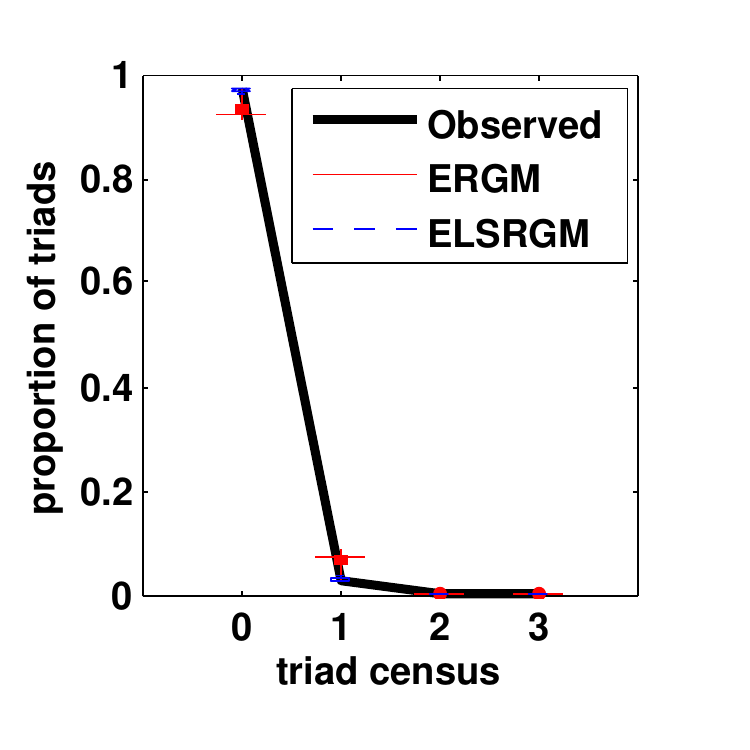}
\end{tabular}
\caption{{\tt Faux-Mesa-High} $205$-node network summaries. The observed statistics are indicated by the solid lines; box plots summarize the statistics for the simulated networks. $ERGM$ plots shows signs of degeneracy effect.}
\label{fig:ergm-elsrgm-faux}
\end{figure*}
Finally, we assess \ELSRGM's performance on a benchmark data set from social network analysis, which is used in testing whether a model can overcome the degeneracy phenomenon. We consider the first matrix ($39$ nodes) of the Kapferer's tailor shop data \cite{Kapferer1972}.  The result of $100$ simulated network summaries shown in Figure \ref{fig:ergm-elsrgm-kaiferer} suggests that the proposed \ELSRGM\  does not suffer from degeneracy. ERGM results were not included for this data set due to severe unstable estimation of MLE using simple specifications considered in our analysis (corresponding to the first type of degeneracy).

\begin{figure*}[!tp]
\centering
\begin{tabular}{cccc}
\hspace{-1cm}\includegraphics[scale=0.49]{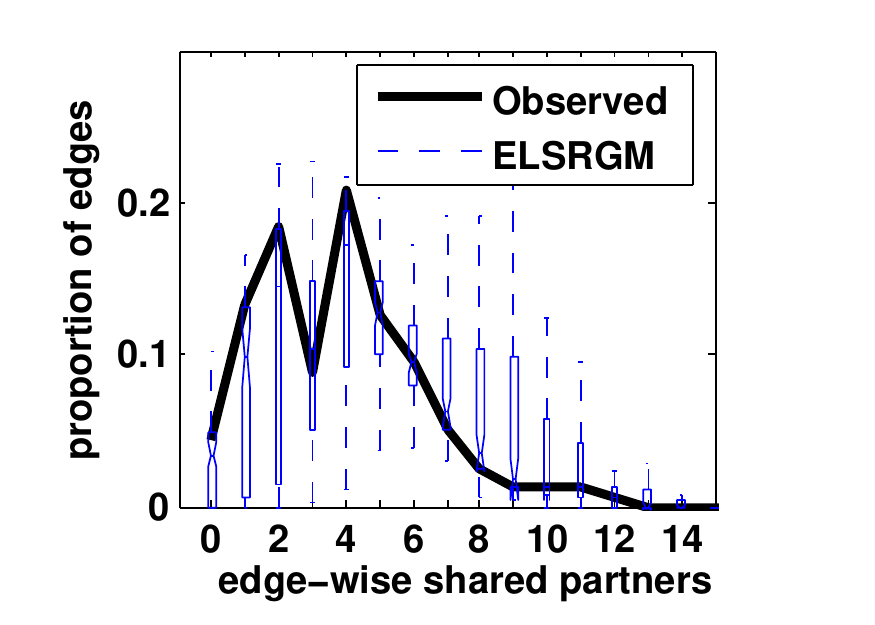}&\hspace{-1cm}\includegraphics[scale=0.49]{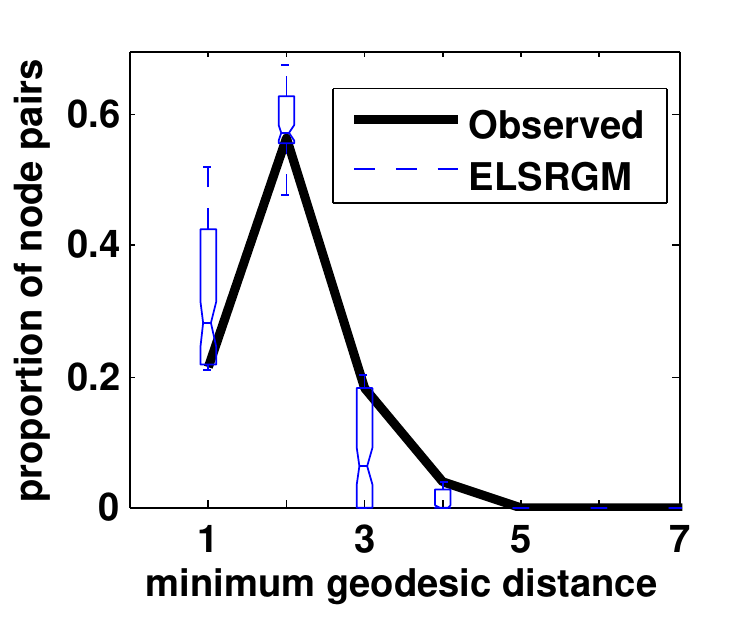} \includegraphics[scale=0.49]{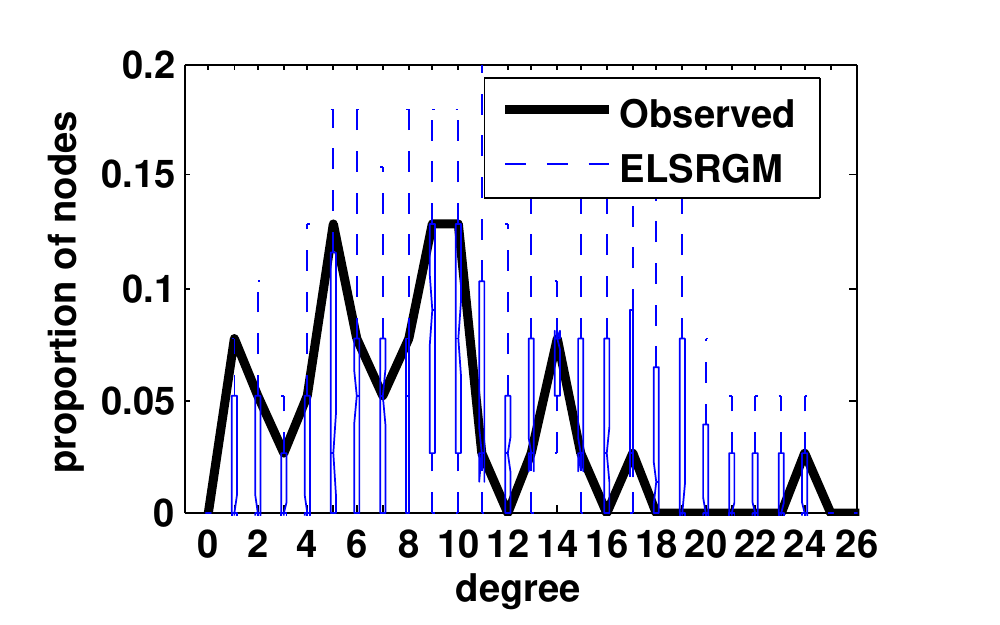}&\hspace{-1cm}\includegraphics[scale=0.49]{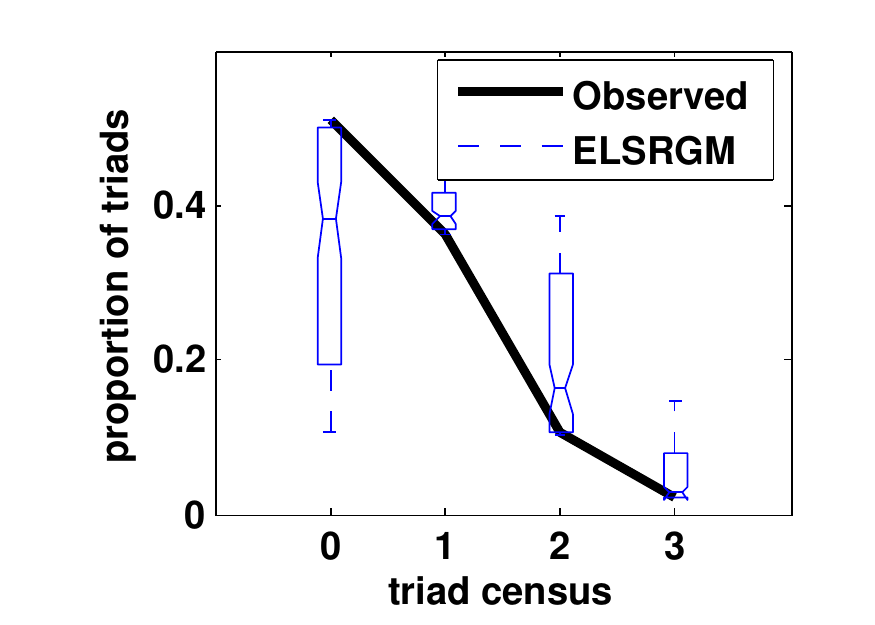}
\end{tabular}
\caption{Kapferer's $39$-node network summaries. The observed statistics are indicated by the solid lines; box plots summarize the statistics for the simulated networks using {\tt ELSRGM}.}
\label{fig:ergm-elsrgm-kaiferer}
\end{figure*}

\begin{figure*}[!p]
\centering
\begin{tabular}{cc}
\includegraphics[scale=0.5]{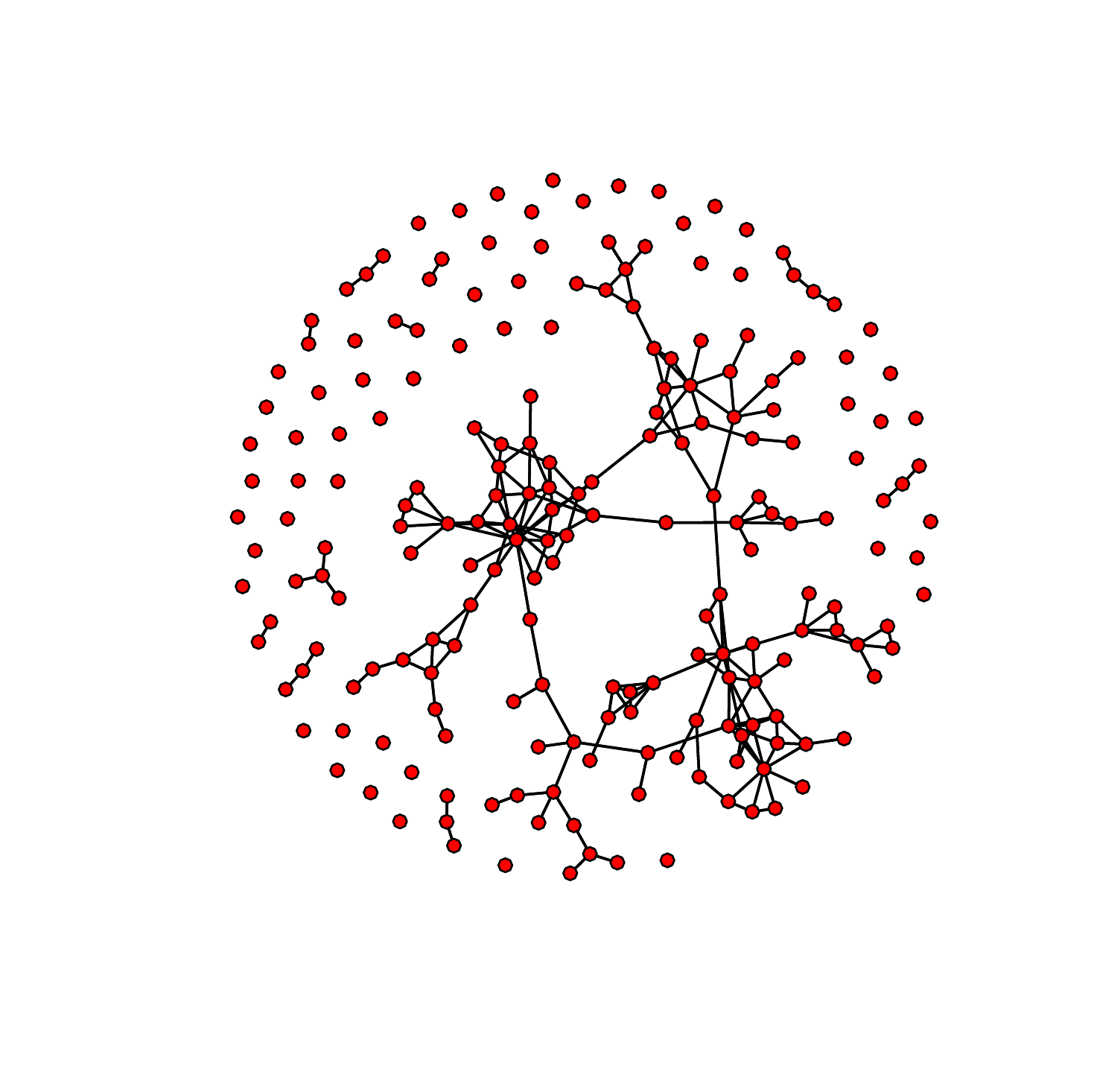} & \includegraphics[scale=0.5]{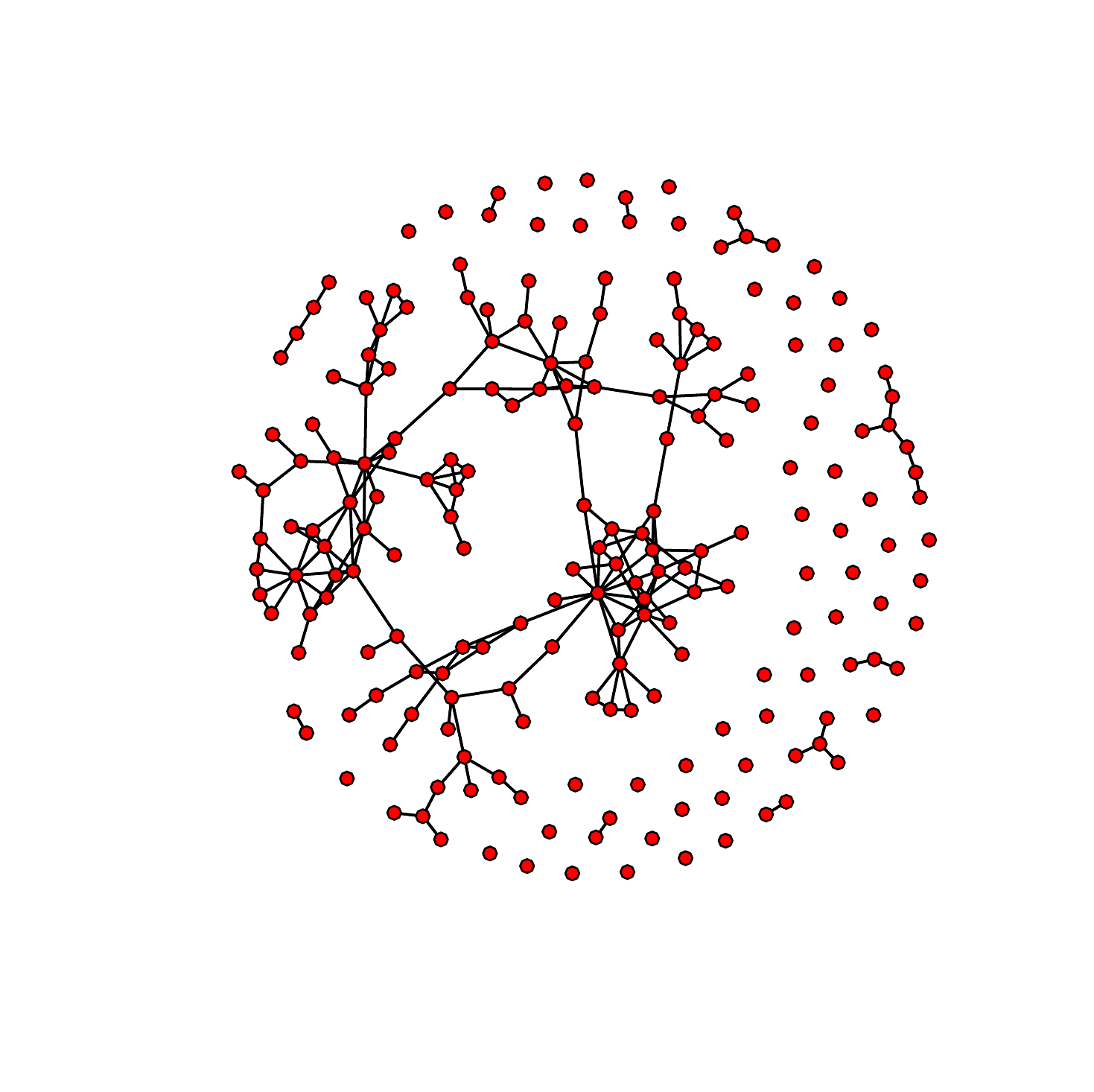}
\end{tabular}
\caption{\textit{left:} {\tt Faux-Mesa-High} $205$-node original network. \textit{right:} {\tt ELSRGM} generated network.}
\end{figure*}
\section{Conclusion and Future Work}\label{sec:conclusion}
In this paper, we investigated the cause of degeneracy in ERGMs, and explained the degeneracy as related to the feature vectors belonging to the relative interior of the convex polytope of feature values.  To correct this issue, we proposed an algorithm that uses spherical features as points on a sphere would belong to the relative boundary of the corresponding convex set rather than its relative interior. Based on this mapping, we outlined a novel model for graphs (\ELSRGM) and an approach to sampling graphs from this model.  In several synthetic and real-world social network, our approach generated graphs with statistics similar to that of the example graphs while not suffering from the issue of degeneracy (unlike ERGMs).

\ELSRGM\ opens up a new class of graph sampling models: those based on spherical features.  In this paper, we made several modeling choices, e.g., von Mises-Fisher distribution for spherical density, kernel approach assigning mass to individual feature vectors; other modeling choices could also lead to models preserving properties of the example graphs, and they need to be investigated.  The insight from the geometric interpretation of the feature vectors realizable as modes may also lead to other types of features (non-spherical) which can lead to models generating realistic graphs.

\section{Acknowledgements}\label{sec:acknowledgements}
The authors thank Okan Ersoy, S.V.N. Vishwanathan, and Richard C. Wilson for helpful
discussions and suggestions. This research was supported
by the NSF Award IIS-0916686. 

\bibliographystyle{abbrv}
\bibliography{elsrgmreport}

\end{document}